\documentclass[11pt,a4paper]{article}
\pdfoutput=1

\usepackage[margin=1in]{geometry}

\usepackage{amsmath}  
\usepackage{amsthm}  
\usepackage{amssymb}  
\usepackage{enumitem}  
\usepackage{filecontents}  
\usepackage{hyperref}
\usepackage{tikz}
\usepackage{xspace}

\hypersetup{colorlinks=true,citecolor=blue, linkcolor=blue, urlcolor=blue}

\newtheorem{thm}{Theorem}

\newtheorem{lem}[thm]{Lemma}
\newtheorem{cor}[thm]{Corollary}
\theoremstyle{definition}
\newtheorem{dfn}[thm]{Definition}

\newtheorem{obs}[thm]{Observation}
\newtheorem*{rem}{Remark}

\newtheorem*{rep@theorem}{\rep@title}
\newcommand{\newreptheorem}[2]{%
	\newenvironment{rep#1}[1]{%
    \def\rep@title{#2 \ref{##1} (restated)}%
		\begin{rep@theorem}\itshape}%
		{\end{rep@theorem}}}
\newreptheorem{thm}{Theorem}
\newreptheorem{lem}{Lemma}
\newreptheorem{prop}{Proposition}
\newreptheorem{cor}{Corollary}

 \let\epsilon=\varepsilon
 \def\ZPM{Z_{\text{PM}}}
 \def\ZNPM{Z_{\text{NPM}}}

\let\op=\oplus  
\let\meet=\wedge  
\let\join=\vee

\let\wh=\widehat

\def\zN{\mathbb{N}}  
\def\zQ{\mathbb{Q}}  
\def\cB{\mathcal{B}}  
\def\cE{\mathcal{E}}
\def\cF{\mathcal{F}}
\def\cG{\mathcal{G}}
\def\cM{\mathcal{M}}
\def\cN{\mathcal{N}}
\def\cP{\mathcal{P}}   
\def\PM{\mathsf{Pin\text{-}MON}}
\def\MON{\mathsf{MON}}
\def\SDP{\mathsf{SDP}}
\newcommand{\IMtwo}{\mbox{$I\!M_2$}\xspace} 
\def\Gud{\ang{\cF,\upf,\downf}_{\om,p}}  
\def\Fzo{\ang{\cF,\delta_0,\delta_1}}  
\def\wte{\mathsf{WtEven3}}

\def\ba{{\bf a}}  
\def\bb{{\bf b}}  
\def\bc{{\bf c}}  
\def\bd{{\bf d}}  
\def\bp{{\bf p}}

\def\bv{{\bf v}}  
\def\bx{{\bf x}}  
\def\by{{\bf y}}

\def\NCSP{{\rm \#CSP}}
\def\NEQ{\mathrm{NEQ}}
\def\LSM{\mathsf{LSM}}
\def\EQ{\mathrm{EQ}}
\def\hol{{\sf Holant}} 
\def\IMP{\mathrm{IMP}}
\def\BIS{{\sf \#BIS}}
\def\SAT{{\sf \#SAT}}
\def\FP{{\sf FP}}
\def\NP{{\sf NP}}
\def\RP{{\sf RP}}
\def\sP{{\sf \#P}}

\let\sse=\subseteq  
\def\ang#1{\langle #1 \rangle}  
\def\vc#1#2{#1 _1\zd #1 _{#2}}  
\def\zd{,\ldots,}  
\def\upf{\mathsf{up}}
\def\downf{\mathsf{down}}
\def\binar#1#2#3#4{
\left(\begin{array}{cc}
#1 & #2\\ #3 & #4
\end{array}\right)
}

\def\ari{\operatorname{arity}}
\def\abs#1{\left| #1 \right|}

\let\al=\alpha  
\let\gm=\gamma  
\let\dl=\delta
\let\ld=\lambda  
\let\om=\omega   
\let\Gm=\Gamma

\title{Boolean approximate counting CSPs with weak conservativity, and implications for ferromagnetic two-spin}
\author{Miriam Backens\thanks{The research leading to these results has received funding from the European Research Council under the European Union's Seventh Framework Programme (FP7/2007--2013) ERC grant agreement no.\ 334828 (Backens and Goldberg) and Horizon 2020 research and innovation programme (grant agreement no.\ 714532, \v{Z}ivn\'{y}). The paper reflects only the authors' views and not the views of the ERC or the European Commission. The European Union is not liable for any use that may be made of the information contained therein.} \and Andrei Bulatov\thanks{Supported by an NSERC Discovery Grant.} \and Leslie Ann Goldberg\footnotemark[1] \and Colin McQuillan \and Stanislav \v{Z}ivn\'{y}\footnotemark[1] \thanks{Supported by a Royal Society University Research Fellowship.}}

\date{15 December 2019}

\begin{document}

\maketitle

\begin{abstract}
 We analyse the complexity of approximate counting constraint satisfactions problems $\NCSP(\cF)$, where $\cF$ is a set of nonnegative rational-valued functions of Boolean variables.
 A complete classification is known in the conservative case, where 
 $\cF$ is assumed to contain arbitrary unary functions. 
 We strengthen this result by fixing any permissive strictly increasing unary function and any permissive strictly decreasing unary function, and adding only those to $\cF$: this is weak conservativity.
 The resulting classification is employed to characterise the complexity of a wide range of two-spin problems, fully classifying the ferromagnetic case.
 In a further weakening of conservativity, we also consider what happens if only the pinning functions are 
assumed to be in~$\cF$ (instead of the two permissive unaries).
 We show that any set of functions 
 for which pinning is not sufficient to recover the two kinds of permissive unaries 
 must either have a very simple range, or must satisfy a certain monotonicity condition. 
 We   exhibit a non-trivial example of a set of functions satisfying the monotonicity condition.
\end{abstract}

\section{Introduction}

A counting constraint satisfaction problem (counting CSP or \NCSP) is parameterised by a finite set $\cF$ of functions taking values in some ring.
An instance $\Omega$ of $\NCSP(\cF)$ consists of a set of variables taking values in some domain $D$, and a set of constraints.
Each constraint is a tuple containing a list of (not necessarily distinct) variables, called the scope, and a constraint function, which is an element of $\cF$ whose arity is equal to the number of variables in the scope.

Any assignment of values to the variables yields a \emph{weight}, which is the product of the resulting values of the constraint functions.
Given the instance $\Omega$, the computational problem is to determine (either exactly or approximately) the 
sum of the weights of all assignments.

Many counting problems can be expressed in the counting CSP framework.
Consider, for example, the problem of counting the number of 2-colourings of a finite graph $G=(V,E)$.
A $2$-colouring is an assignment from $V$ to $\{0,1\}$
which has the property that any two vertices connected by an edge are assigned different values.
This can be expressed as an instance $\Omega$ of $\NCSP(\{f\})$, where $f$ is the symmetric binary function satisfying $f(0,1)=f(1,0)=1$ and $f(0,0)=f(1,1)=0$.
The variables of the instance correspond to the vertices of the graph and the constraints correspond to the edges.
An assignment  of values in $\{0,1\}$ to the variables has weight 1 if it corresponds to a valid 2-colouring, and weight 0 otherwise.
Thus the sum of the weights of all assignments is exactly the number of 2-colourings of $G$.

Counting CSPs are closely related to certain problems arising in statistical physics.
Each variable can be thought of as an object which can be in one of several states.
Adjacent objects interact and the strength of such an interaction is captured by a constraint
function. Thus, an assignment associates states with objects.
The sum of the weights of all assignments, denoted $Z(\Omega)$, is called the  \emph{partition function} of the physical system.

Throughout this paper, we will consider Boolean counting CSPs,
which are counting CSPs in which the variables take values from the domain $D=\{0,1\}$.
Constraint functions will be assumed to take nonnegative rational values.
The set of all arity-$k$ nonnegative rational-valued functions of Boolean inputs is denoted $\cB_k$, and we write $\cB=\bigcup_{k\in\zN}\cB_k$ for the set of nonnegative rational-valued functions of Boolean inputs with arbitrary arity.

The complexity of exactly solving Boolean counting CSPs is fully classified, even 
when the constraint functions are allowed to take algebraic complex values \cite{cai_complexity_2014}.
This classification takes the form of a dichotomy: if $\cF$ is a subset of one of two specific families of functions, the problem is in \FP; otherwise it is \sP-hard.
When the ranges of constraint functions in~$\cF$ are restricted to be nonnegative rational values
(or even algebraic real values),   there is only one tractable family, known as \emph{product-type functions} and denoted $\cN$ (see Definition~\ref{def:cN}).

In this paper, we consider the complexity of \emph{approximately} solving counting CSPs. We 
classify constraint families~$\cF$ according to whether or not there is a fully polynomial-time randomised approximation scheme (FPRAS) for  the problem $\NCSP(\cF)$.

If all the constraints in~$\cF$ are Boolean functions, i.e.\ functions in $\cB$ whose range is $\{0,1\}$, 
then the approximation problem $\NCSP(\cF)$  is fully classified \cite{Dyer10:approximation}. We state the  precise classification of \cite{Dyer10:approximation} as Theorem~\ref{thm:complexity_Boolean} of this paper.
Informally,
the classification takes the form of a trichotomy, separating problems into ones that are in \FP, ones that are equivalent to the problem of counting independent sets in a bipartite graph (denoted \BIS), and ones that do not have an FPRAS unless $\NP=\RP$.
The equivalence is under approximation-preserving reductions (AP-reductions); in the following, we write $A\leq_{AP} B$ if $A$ can be reduced to $B$ under AP-reductions.

There is also a \emph{conservative} classification for nonnegative efficiently-computable real-valued functions, where ``conservative'' means that 
$\cF$ is assumed to  contain  arbitrary unary functions   \cite{bulatov_expressibility_2013}.
This classification straightforwardly restricts to 
the case in which constraint functions take
non-negative rational values \cite{chen_complexity_2015}.
The restriction is stated in this paper as Theorem~\ref{thm:complexity_conservative}.
It uses the class of log-supermodular functions, denoted $\LSM$.
A $k$-ary function $f$ is log-supermodular if $f(\bx\vee\by)f(\bx\wedge\by)\geq f(\bx)f(\by)$ for all $\bx,\by\in\{0,1\}^k$, where $\vee$ and $\wedge$ are applied bit-wise.
The classification shows that $\NCSP(\cF)$ is at least as hard as \BIS\ if $\cF\nsubseteq\cN$.
If $\cF\nsubseteq\cN$
and $\cF\nsubseteq \LSM$ then
it is (presumably) even harder -- it is as hard as counting the satisfying assignments of a Boolean formula (so there is no FPRAS unless $\NP=\RP$).
The paper \cite{bulatov_expressibility_2013}
also implies a \BIS-easiness result for the case when $\cF \subseteq \LSM$
and all constraint functions have arity at most three.

In this paper, we give a complexity classification for counting CSPs under a significantly weaker conservativity assumption: instead of  adding arbitrary unary functions to~$\cF$, we fix \emph{one} strictly increasing 
permissive unary function
and \emph{one} strictly decreasing permissive unary function, and  we add (only) these to~$\cF$.

\newcommand{\stateupdown}[4]{
 Let $\upf$ be a permissive unary strictly increasing function, let $\downf$ be a permissive unary strictly decreasing function, and let $\cF\sse\cB$.
 Then the following properties hold.
 \begin{enumerate}
  \item #1
   If $\cF\sse\cN$, then, for any finite subset $S$ of $\cF$, $\NCSP(S\cup\{\upf,\downf\})$ is in \FP.
  \item
   Otherwise, if $\cF\sse\LSM$, then
    \begin{enumerate}
     \item #2 there is a finite subset $S$ of $\cF$ such that $\BIS \leq_{AP} \NCSP(S\cup\{\upf,\downf\})$, and
     \item #3 for every finite subset $S$ of $\cF$ such that all functions $f\in S$ have arity at most 2, $\NCSP(S\cup\{\upf,\downf\})\leq_{AP}\BIS$.
    \end{enumerate}
  \item #4
   Otherwise, there is a finite subset $S$ of $\cF$ such that $\NCSP(S\cup\{\upf,\downf\})$ does not have an FPRAS unless $\NP=\RP$.
 \end{enumerate}
}
\begin{thm}\label{thm:up-down}
 \stateupdown{\label{p:cN_FP}}{\label{p:lsm_BIS-hard}}{\label{p:lsm_BIS-easy}}{\label{p:non-lsm_hard}}
\end{thm}

Our main application of Theorem~\ref{thm:up-down} 
is a characterisation of  the complexity of two-spin problems.
A two-spin problem corresponds to the problem  $\NCSP(\{f\})$ 
where the constraint function~$f$ is in~$\cB_2$.
These problems arise in statistical physics. For example, the problem of 
computing the partition function of the Ising model 
is a two-spin problem where, for some value~$\beta$,
$f(0,0)=f(1,1)=\beta$ and $f(0,1)=f(1,0)=1$.

Our application   requires the following definitions. We say that a binary function $f$ is monotone if $f(0,0) \leq f(0,1) \leq f(1,1)$ and $f(0,0) \leq f(1,0) \leq f(1,1)$.
The Fourier coefficients of $f$ are 
defined by $\wh{f}_{xy} = \frac{1}{4} \sum_{p,q\in\{0,1\}} (-1)^{px+qy} f(p,q)$ for all $x,y \in \{0,1\}$.

Let $\EQ$ be the binary equality function
defined by $\EQ(0,0)=\EQ(1,1)=1$ and $\EQ(0,1)=\EQ(1,0)=0$.
Let  $\NEQ$ the binary disequality function 
defined by $\NEQ(0,0)=\NEQ(1,1)=0$ and $\NEQ(0,1)=\NEQ(1,0)=1$.
A binary function $f$ is \emph{log-modular} if $f(0,1)f(1,0)=f(0,0)f(1,1)$.
A binary function $f$ is called \emph{trivial} if it is log-modular  or 
there is a unary function $g\in \cB_1$ such that
$f(x,y) = g(x) \EQ(x,y)$ or 
$f(x,y)=g(x)\NEQ(x,y)$. 
A binary function $f$ is  log-supermodular if and only if it is \emph{ferromagnetic}, 
which means that $f(0,0) f(1,1) \geq f(0,1)f(1,0)$.
Our classification theorem is as follows.

\newcommand{\statetwospin}[5]{
 Let $f\in\cB_2$.
 \begin{enumerate}
  \item #1 If $f$ is trivial, then $\NCSP(\{f\})$ is in \FP.
  \item #2 Otherwise, if $f$ is ferromagnetic:
     \begin{enumerate}
    \item #3 If $\wh{f}_{01}\wh{f}_{10}<0$, then $\NCSP(\{f\})$ is equivalent to \BIS\ under AP-reductions.
    \item #4 Otherwise, $\NCSP(\{f\})$ has an FPRAS.
   \end{enumerate}
  \item #5 Otherwise, if both $f(x,y)$ and $f(1-x,1-y)$ are non-monotone, then $\NCSP(\{f\})$ does not have an FPRAS unless $\NP=\RP$. 
 \end{enumerate}
}
\begin{thm}\label{thm:two-spin}
 \statetwospin{\label{p:trivial}}{\label{p:lsm}}{\label{p:lsm_BIS-compl}}{\label{p:lsm_FPRAS}}{\label{p:non-monotone}}
\end{thm}
 
The classification in Theorem~\ref{thm:two-spin}  is not exhaustive: $\NCSP(\{f\})$ is now fully classified if $f$ is   ferromagnetic.
For   anti-ferromagnetic functions~$f$, the complexity of $\NCSP(\{f\})$ is still open, except in the 
doubly non-monotone case 
(where both $f(x,y)$ and $f(1-x,1-y)$ are non-monotone)
and in the symmetric case where $f(x,y) = f(y,x)$.
The symmetric case has been resolved with a classification into 
situations where 
$\NCSP(\{f\})$ has 
a fully polynomial-time approximation scheme (FPTAS)
and 
situations where $\NCSP(\{f\})$ does not have an FPRAS unless $\NP=\RP$.
More details are given in Section~\ref{sec:existing} however we mention here that
the known classification depends on universal uniqueness and there is no simple closed form criterion, cf.\ Theorems~\ref{thm:uniqueness_FPTAS} and \ref{thm:non-uniqueness_hard}, which are taken from \cite{Li11:correlation}.

Theorem~\ref{thm:up-down} relaxes the conservativity assumption in known 
counting CSP classifications by 
adding only two permissive unaries to the set~$\cF$ of constraint functions.
Pushing this idea even further, we consider what happens if we allow the pinning functions 
$\delta_0$ and $\delta_1$ defined by
$\delta_0(0)=1, \delta_0(1)=0$ and $\delta_1(0)=0, \delta_1(1)=1$ instead of 
the permissive unaries
$\upf$ and $\downf$.

Given a strictly decreasing permissive unary function $\downf$, $\delta_0$ can be \emph{realised}: this means that the effect of a constraint using the function $\delta_0$ can be (approximately) simulated by some combination of constraints using $\downf$.
Similarly, $\delta_1$ can be realised using $\upf$.
This notion of realisation is formalised using the theory of functional clones and $(\om,p)$-clones in Section~\ref{sec:pps}.
It implies that allowing only the pinning functions instead of allowing both kinds of permissive unaries is a   weaker assumption, though not necessarily a strictly weaker one.

For many sets of functions $\cF\sse\cB$, adding pinning is sufficient to realise two permissive unaries with the desired properties: then the complexity classification of $\NCSP(\cF\cup\{\delta_0,\delta_1\})$ follows from Theorem~\ref{thm:up-down}.
If pinning does not yield both kinds of permissive unaries, we show that either every function $f\in\cF$ has range $\{0,r_f\}$ for some nonnegative rational $r_f$, or all functions in $\cF$ satisfy a certain monotonicity condition.
These results form Theorem~\ref{thm:set_pinning}.

In Section~\ref{sec:4.2}, the goal is to identify a large functional clone that does not contain both a strictly increasing permissive unary function and a strictly decreasing permissive unary function.
In other words, we are looking for a set of functions which provably does not allow both a strictly increasing permissive unary function and a strictly decreasing permissive unary function to be realised.
The set of monotone functions fits the bill, but in some sense it is a trivial solution since it does not contain both~$\delta_0$ and~$\delta_1$.
Theorem~\ref{thm:mon-pm} identifies a functional clone that is strictly larger than the set of monotone functions and contains $\delta_0$ and~$\delta_1$ but still does not contain both a strictly increasing permissive unary function and a strictly decreasing permissive unary function.
This shows that the monotonicity property in the pinning classification can be satisfied in a nontrivial way.

\section{Definitions and preliminaries}
\label{s:definitions}

Throughout this paper, we consider nonnegative rational-valued pseudo-Boolean functions, i.e.\ functions from $\{0,1\}^k$ to $\zQ_{\geq 0}$.
We write $\cB_k$ for the set of all nonnegative rational-valued pseudo-Boolean functions of arity $k$, and $\cB=\bigcup_{k\in\zN}\cB_k$.

A function $f\in\cB_k$ is \emph{permissive} if $f(\bx)>0$ for all $\bx\in\{0,1\}^k$, i.e.\ all its values are non-zero.
The set of permissive unary strictly decreasing functions and the set of permissive unary strictly increasing functions will be of particular interest; we denote them by
\begin{align*}
 \cB^{>}_1 &:= \{ f\in\cB_1: f(0)>f(1)>0 \} \quad \text{and} \\
 \cB^{<}_1 &:= \{ f\in\cB_1: 0<f(0)<f(1) \}.
\end{align*}
These two sets differ from the sets $\cB^{\upf,p}_1$ and $\cB^{\downf,p}_1$ defined in \cite{bulatov_expressibility_2013} as the latter do not require strictness or permissiveness, and they allow polynomial-time computable real values
rather than just rational values.
We will also sometimes require normalised unary functions, thus we define $\cB^{>,\mathrm{n}}_1 := \left\{ f\in\cB^{>}_1: f(0)=1 \right\}$ and $\cB^{<,\mathrm{n}}_1 := \left\{ f\in\cB^{<}_1: f(1)=1 \right\}$.

The \emph{relation underlying a function} $f\in\cB_k$ 
(also called the ``support of $f$'')
is defined as
\[
 R_f = \{ (\vc{x}{k}) : f(\vc{x}{k})\neq 0\}.
\] 

A relation is \emph{affine} if it contains exactly the tuples specified by a set of linear equations over $\operatorname{GF}(2)$.
If the underlying relation of $f$ is affine, we say that $f$ \emph{has affine support}.
A function is \emph{pure affine} if has affine support and its range is $\{0,r\}$ for some $r>0$ \cite{dyer_complexity_2009}.
We extend this definition to say that any function is \emph{pure} if its range is $\{0,r\}$ for some $r>0$, regardless of its support.

A function $f\in\cB_k$ is \emph{log-supermodular} (or \emph{lsm}) if $f(\bx\vee\by)f(\bx\wedge\by) \geq f(\bx)f(\by)$ for all $\bx,\by\in\{0,1\}^k$, where $\vee$ and $\wedge$ are applied bit-wise.
The class of all lsm functions of any arity is denoted $\LSM$.
A $k$-ary function $f$ is \emph{log-modular} if $f(\bx\vee\by)f(\bx\wedge\by) = f(\bx)f(\by)$ for all $\bx,\by\in\{0,1\}^k$.
It is straightforward to check that all unary functions are both lsm and log-modular.

For any positive integer~$n$, we write $[n] :=\{1\zd n\}$.

Let $\delta_0$ and $\delta_1$ be the unary functions satisfying $\delta_0(0)=1, \delta_0(1)=0$ and $\delta_1(0)=0, \delta_1(1)=1$; these are often called the \emph{pinning functions}.
If $f$ is a function in $\cB_k$ with $k\geq 2$, then a \emph{2-pinning} of $f$ is a binary function $g$ that arises from $f$ by pinning all but two of the variables to a fixed value.
Formally, $g$ is of the form
\[
 g(x_p,x_q) = \sum_{(x_{i_1},\ldots,x_{i_{k-2}}) \in \{0,1\}^{k-2}} f(x_1\zd x_k) \prod_{i\in[k]\setminus\{p,q\}} \delta_{a_i}(x_i),
\]
where $p$ and $q$ are distinct indices in $[k]$,
$\{i_1,\ldots,i_{k-2}\} = [k] \setminus \{p,q\}$, 
and for each $i\in [k] \setminus \{p,q\}$,  $a_i$ is a value in $\{0,1\}$. 

A function $f\in\cB_k$ is \emph{monotone} if for any $\ba,\bb\in\{0,1\}^k$ with $\ba\leq\bb$ we have $f(\ba)\leq f(\bb)$.

A function is \emph{monotone on its support} if for any $\ba,\bb\in R_f$ with $\ba\leq\bb$ we have $f(\ba)\leq f(\bb)$.
All monotone functions are also monotone on their support, but the latter set is bigger: a function $f\in\cB_k$ that is monotone on its support may have inputs $\ba,\bb\in\{0,1\}^k$ such that $\ba\leq\bb$ and $f(\ba)>f(\bb)$, as long as $f(\bb)=0$.
For example, $\delta_0$ is trivially monotone on its support, but it is not monotone.

Let $\bar{x}=1-x$ for $x\in\{0,1\}$.
The \emph{bit-flip} of a function $f\in\cB_k$ is the function $\bar{f}(\vc{x}{k}) = f(\vc{\bar{x}}{k})$.
The bit-flip of a set $\cF\sse\cB$ is $\bar{\cF}:=\{\bar{f}:f\in\cF\}$.

\subsection{Binary functions}\label{sec:binfn}

Much of this paper is concerned with binary functions, it is thus useful to introduce specific notation and results.
A binary function is said to be \emph{ferromagnetic} if it is log-supermodular
and it is \emph{anti-ferromagnetic} otherwise.
We often write a binary function as a $2\times 2$ matrix:
\[
 f(x,y) = \begin{pmatrix} f(0,0) & f(0,1) \\ f(1,0) & f(1,1) \end{pmatrix}.
\]

\begin{obs}\label{obs:binary_lsm}
 A binary function $f$ is log-supermodular (or, equivalently, it is ferromagnetic) if $f(0,0)f(1,1)\geq f(0,1)f(1,0)$.
\end{obs}

The binary equality function
 $\EQ$  is
defined by $\EQ(0,0)=\EQ(1,1)=1$ and $\EQ(0,1)=\EQ(1,0)=0$
and the binary disequality function
 $\NEQ$ is
defined by $\NEQ(0,0)=\NEQ(1,1)=0$ and $\NEQ(0,1)=\NEQ(1,0)=1$. 
A binary function $f$ is called \emph{trivial} if it is log-modular  or 
there is a unary function $g\in \cB_1$ such that
$f(x,y) = g(x) \EQ(x,y)$ or 
$f(x,y)=g(x)\NEQ(x,y)$.
It is easy to see that a binary function $f$ is log-modular if and only if $f(x,y)=g(x)h(y)$ for some unary functions $g$ and $h$.
We  will use the following observation.

\begin{obs}\label{obs:permissive_binary_trivial}
 A permissive binary function is trivial if and only if it is log-modular.
\end{obs}

A binary function $f$ is \emph{symmetric} if $f(0,1)=f(1,0)$ and it is an \emph{Ising function} if it depends only on the parity of its input, i.e.\ $f(0,1)=f(1,0)$ and $f(0,0)=f(1,1)$.

It is often easier to work with symmetric binary functions than with general ones.
The following lemma gives some properties of different methods for symmetrising a given function.

\begin{lem}\label{lem:symmetric}
 Let $f$ be a nontrivial binary function.
 \begin{enumerate}
  \item\label{p:symmetric_non-lsm}
   If $f$ is non-lsm, then $f'(x,y)=f(x,y)f(y,x)$ is nontrivial, symmetric, and non-lsm.
  \item\label{p:symmetric_lsm}
   If $f$ is lsm, then $f''(x,y)=\sum_{z\in\{0,1\}} f(x,z)f(y,z)$ is nontrivial, symmetric, and lsm.
  \item\label{p:Ising}
   If $f$ is lsm and Ising, then $f'''(x,y)=\sum_{z\in\{0,1\}} f(x,z)f(y,z)\upf(z)$ is nontrivial, symmetric, lsm, and not Ising, where $\upf$ is any strictly increasing permissive unary function.
 \end{enumerate}
\end{lem}
\begin{proof}
 Throughout, we write
 \[
  f(x,y) = \begin{pmatrix} a&b\\c&d \end{pmatrix}.
 \]

 For Property~\ref{p:symmetric_non-lsm}, suppose $f$ is nontrivial and non-lsm, i.e.\ $ad-bc<0$ and $a,d$ are not both zero.
 Let $f'(x,y)=f(x,y)f(y,x)$, i.e.\ $f'(0,0)=a^2$, $f'(0,1)=bc=f'(1,0)$, and $f'(1,1)=d^2$.
 It is straightforward to see that $f'$ is symmetric.
 As $b,c>0$ and $a,d$ are not both zero, $f'$ does not have the form $g(x)\EQ(x,y)$ or $g(x)\NEQ(x,y)$ for any unary function $g$.
 Furthermore, $ad<bc$ and $a,b,c,d\geq 0$ implies that $a^2d^2<b^2c^2=(bc)^2$, so $f'(0,0)f'(1,1)<f'(0,1)f'(1,0)$.
 This shows that $f'$ is not log-modular and therefore nontrivial.
 Additionally, it also shows that $f'$ is non-lsm, concluding the proof of the property.
 
 For Property~\ref{p:symmetric_lsm}, suppose $f$ is nontrivial and lsm, i.e.\ $ad-bc>0$ and $b,c$ are not both zero.
 Let
 \[
  f''(x,y)=\sum_{z\in\{0,1\}} f(x,z)f(y,z) = \begin{pmatrix} a^2+b^2 & ac+bd \\ ac+bd & c^2+d^2 \end{pmatrix}.
 \]
 This can easily be seen to be symmetric.
 Since at most one of $a,b,c,d$ is zero, $f''$ is permissive.
 Finally, $f''(0,0)f''(1,1)-f''(0,1)f''(1,0) = (ad-bc)^2>0$, so (using Observation~\ref{obs:permissive_binary_trivial}) $f''$ is nontrivial and lsm.
 
 For Property~\ref{p:Ising}, suppose $f$ is nontrivial, lsm and Ising, i.e.\ $ad-bc>0$, $a=d$, and $b=c$, with $a,b,c,d>0$.
 We can thus write
 \[
  f(x,y) = \begin{pmatrix} a&b\\b&a \end{pmatrix};
 \]
 the property of $f$ being both nontrivial and lsm becomes $a^2-b^2>0$.
 Now, let $u_0:=\upf(0)$ and $u_1:=\upf(1)$; these values satisfy $0<u_0<u_1$.
 Then
 \[
  f'''(x,y) = \sum_{z\in\{0,1\}} f(x,z)f(y,z)\upf(z) = \binar{a^2 u_0 + b^2 u_1}{ab(u_0+u_1)}{ab(u_0+u_1)}{a^2 u_1 + b^2 u_0}.
 \]
 Since $a,b,u_0,u_1$ are all positive, $f'''$ is permissive.
 It is also clearly symmetric.
 Furthermore,
 \[
  f'''(0,0)f'''(1,1)-f'''(0,1)f'''(1,0) = (a^2-b^2)^2 u_0 u_1 > 0.
 \]
 Thus, using Observation~\ref{obs:permissive_binary_trivial}, $f'''$ is nontrivial and lsm.
 The function $f'''$ is Ising if $a^2 u_0 + b^2 u_1 = a^2 u_1 + b^2 u_0$ or, equivalently, if $(a^2-b^2)(u_0-u_1)=0$.
 But $f$ being nontrivial implies that $a^2-b^2\neq 0$ and $\upf$ being strictly increasing implies that $u_0\neq u_1$, hence this equality is never satisfied.
 That means $f'''$ cannot be Ising, so $f'''$ has all the desired properties.
\end{proof}

\subsection{Fourier transforms}
\label{s:Fourier}

The Fourier transform of a $k$-ary function $f\in\cB_k$ is given by
\[
 \wh{f}(\vc{x}{k}) = \frac{1}{2^k} \sum_{\vc{p}{k}\in\{0,1\}} (-1)^{p_1x_1 + \ldots + p_kx_k} f(\vc{p}{k}).
\]
The values of $\wh{f}$ are called the Fourier coefficients of $f$.
The following fact is well-known. See, e.g., \cite[Equation 2.1]{de_wolf_brief_2008}.

\begin{obs} \label{obs:Fourier_inverse}
 For any $f\in\cB_k$,
 \[
  f(\vc{x}{k}) = \sum_{\vc{p}{k}\in\{0,1\}} (-1)^{p_1x_1 + \ldots + p_kx_k} \wh{f}(\vc{p}{k}).
 \]
\end{obs}

We denote by $\cP$ the set of all nonnegative functions whose Fourier coefficients are also nonnegative:
\[
 \cP = \left\{f\in\cB : \wh{f}(\bx)\geq 0 \text{ for all } \bx\in\{0,1\}^{\ari(f)}\right\},
\]
where $\ari(f)$ denotes the arity of~$f$.
If $f$ is binary, with
\[
 f(x,y) = \begin{pmatrix} a&b\\c&d \end{pmatrix},
\]
its Fourier coefficients are sometimes written as
\begin{align*}
 \wh{f}_{00} &:= \wh{f}(0,0) = \frac{1}{4}(a+b+c+d) \\
 \wh{f}_{01} &:= \wh{f}(0,1) = \frac{1}{4}(a-b+c-d) \\
 \wh{f}_{10} &:= \wh{f}(1,0) = \frac{1}{4}(a+b-c-d) \\
 \wh{f}_{11} &:= \wh{f}(1,1) = \frac{1}{4}(a-b-c+d).
\end{align*}

\begin{obs}\label{obs:Fourier_bit_flip}
 Let $f'(x,y)=\bar{f}(x,y)=f(1-x,1-y)$, where we are renaming the function temporarily in order to avoid clumsy notation.
 Then $\wh{f'}_{01}=-\wh{f}_{01}$ and $\wh{f'}_{10}=-\wh{f}_{10}$, while $\wh{f'}_{00}=\wh{f}_{00}$ and $\wh{f'}_{11}=\wh{f}_{11}$.
 In other words, a bit-flip changes the signs of $\wh{f}_{01}$ and $\wh{f}_{10}$ while leaving $\wh{f}_{00}$ and $\wh{f}_{11}$ unaffected.
\end{obs}

\begin{obs}\label{obs:Fourier_swap}
 Let $f''(x,y)=f(y,x)$.
 Then $\wh{f''}_{01}=\wh{f}_{10}$ and $\wh{f''}_{10}=\wh{f}_{01}$, while again $\wh{f''}_{00}=\wh{f}_{00}$ and $\wh{f''}_{11}=\wh{f}_{11}$.
 Hence, swapping the variables swaps the values of $\wh{f}_{01}$ and $\wh{f}_{10}$.
\end{obs}

These properties will be useful when considering two-spin problems in Section \ref{s:two-spin}.

\begin{lem}\label{lem:cP_000}
 Let $f\in\cP$ with arity $k$.
 Then $f(\vc{x}{k})\leq f(0\zd 0)$ for any $\vc{x}{k}\in\{0,1\}$.
\end{lem}
\begin{proof}
 In the following, we write $\bp\cdot\bx=p_1x_1 \oplus\ldots\oplus p_kx_k$, where $\bx=(\vc{x}{k}),\bp=(\vc{p}{k})\in\{0,1\}^k$.
 By  Observation~\ref{obs:Fourier_inverse}, for any $f\in\cB_k$,
 \begin{equation}\label{eq:Fourier_decomposition}
  f(\bx) = \sum_{\bp\in\{0,1\}^k} (-1)^{\bp\cdot\bx} \wh{f}(\bp).
 \end{equation}
 For any $\bx\in\{0,1\}^k$, let $S_{f;\bx}$ be the set of all $k$-bit strings $\bp$ such that $\wh{f}(\bp)$ appears with coefficient $-1$ in \eqref{eq:Fourier_decomposition}:
 \[
  S_{f;\bx}= \left\{\bp\in\{0,1\}^k : \bp\cdot\bx = 1 \right\}.
 \]
 Note that if $\bx$ is the all-zeroes string, then $\bp\cdot\bx=0$ for all $\bp\in\{0,1\}$.
 This implies that $S_{f;0\zd 0}=\emptyset$ and $f(0\zd 0) = \sum_{\bp\in\{0,1\}^k} \wh{f}(\bp)$.
 Thus, \eqref{eq:Fourier_decomposition} can be rewritten as
 \[
  f(\bx) = \sum_{\bp\in\{0,1\}^k} \wh{f}(\bp) - 2\left(\sum_{\bp\in S_{f;\bx}} \wh{f}(\bp)\right) = f(0\zd 0) - 2\left(\sum_{\bp\in S_{f;\bx}} \wh{f}(\bp)\right) \leq f(0\zd 0),
 \]
 where the inequality holds because $f\in\cP$ implies that all terms in the sum are nonnegative. 
\end{proof}

\subsection{Relational clones, functional clones, and \texorpdfstring{$\omega$}{omega}-clones}\label{sec:pps}

Functional clones, $\omega$-clones and $(\omega,p)$-clones are sets of functions that are closed under certain operations.
These operations are listed in Lemma~\ref{lem:closure} for functional clones; $\omega$-clones and $(\omega,p)$-clones additionally have different ways of taking limits.
Their definitions build on the well-established theory of relational clones.
To simplify our notation, we take our definitions from~\cite{bulatov_functional_2017,bulatov_expressibility_2013} though relational clones have a longer history \cite{PK}.  
In the following, let $\Gm$ be a set of relations and let $V=\{\vc{v}{n}\}$ be a set of variables.

\begin{dfn}[\cite{bulatov_expressibility_2013}]
 A \emph{primitive positive formula} (pp-formula) over $\Gm$ in variables $V$ is a formula of the form
 \[
  \exists v_{n+1}\ldots v_{n+m} \bigwedge_{i} \varphi_i,
 \]
 where each atomic formula $\varphi_i$ is either a relation $R$ from $\Gm$ or the equality relation, applied to some of the variables in $V'=\{\vc{v}{n+m}\}$.
\end{dfn}

\begin{dfn}[\cite{bulatov_expressibility_2013}]
 The \emph{relational clone} (or co-clone) $\ang{\Gm}_R$ is the set of all relations expressible as pp-formulas over $\Gm$.
\end{dfn}

The definition of a functional clone is similar, we again follow the explanation in \cite{bulatov_expressibility_2013}.
Note that functional clones, $\om$-clones and $(\om,p)$-clones were originally defined for nonnegative real-valued functions.
The definitions can be restricted to nonnegative rational-valued functions as explained just after Definition~\ref{def:pps-omega}.

Let $\cF\sse\cB$ be a set of functions and $V=\{\vc{v}{n}\}$ a set of variables.
In the context of functions rather than relations, an atomic formula $\varphi=g(v_{i_1}\zd v_{i_k})$ consists of a function $g\in\cF$ and a scope $(v_{i_1}\zd v_{i_k})\in V^k$, where $k=\ari(g)$.
The scope may contain repeated variables.
Given an assignment $\bx:V\to\{0,1\}$, the atomic formula $\varphi$ specifies a function $f_\varphi:\{0,1\}^n\to\zQ_{\geq 0}$ given by
\[
 f_\varphi(\bx) = g(x_{i_1}\zd x_{i_k}),
\]
where $x_j=\bx(v_j)$ for all $j\in [n]$.

\begin{dfn}\label{def:pps}
 A \emph{primitive product summation formula} (pps-formula) in variables $V$ over $\cF$ has the form
 \[
  \psi = \sum_{v_{n+1},\ldots,v_{n+m}} \prod_{j=1}^s \varphi_j,
 \]
 where $\varphi_j$ are all atomic formulas over $\cF$ in the variables $V'=\{\vc{v}{n+m}\}$.
 The variables in $V$ are called \emph{free variables}, those in $V'\setminus V$ are called \emph{bound variables}.
\end{dfn}

The pps-formula $\psi$ represents a function $f_\psi:\{0,1\}^n\to\zQ_{\geq 0}$ given by
\[
 f_\psi(\bx) = \sum_{\by\in\{0,1\}^m} \prod_{j=1}^s f_{\varphi_j}(\bx,\by),
\]
where $\bx$ is an assignment $V\to\{0,1\}$ and $\by$ is an assignment $V'\setminus V\to\{0,1\}$.
If $f$ is represented by some pps-formula $\psi$ over $\cF$, it is said to be \emph{pps-definable} over $\cF$.

\begin{dfn}[\cite{bulatov_expressibility_2013}]
 The \emph{functional clone} generated by $\cF$ is the set of all functions in $\cB$ that can be represented by a pps-formula over $\cF\cup\{\EQ\}$.
 It is denoted by $\ang{\cF}$.
\end{dfn}

There is a another perspective on functional clones \cite{bulatov_functional_2017}.
In the following, let $f\in\cB_k$.
We say a $(k+1)$-ary function $h$ arises from $f$ by \emph{introduction of a fictitious argument} if $h(\vc{x}{k+1})=f(\vc{x}{k})$ for all $\vc{x}{k}\in\{0,1\}$.
Let $g\in\cB_k$, then the \emph{product} of $f$ and $g$ is the function $h$ satisfying $h(\vc{x}{k}) = f(\vc{x}{k}) g(\vc{x}{k})$.
The function $h$ resulting from $f$ by a \emph{permutation of the arguments} $\pi:[k]\to[k]$ is $h(\vc{x}{k})=f(x_{\pi(1)}\zd x_{\pi(k)})$.
Furthermore, a $(k-1)$-ary function $h$ arises from $f$ by \emph{summation} if $h(\vc{x}{k-1})=\sum_{x_k\in\{0,1\}} f(\vc{x}{k})$.

\begin{lem}[{\cite[Section~1.1]{bulatov_functional_2017}}]\label{lem:closure}
 For any $\cF\sse\cB$, $\ang{\cF}$ is the closure of $\cF\cup \{\EQ\}$ under introduction of fictitious arguments, products, permutation of arguments, and summation.
\end{lem}

We adopt the shorthand notation from \cite{bulatov_expressibility_2013}, so if $\cF_1,\ldots,\cF_j$ are sets of functions and $g_1,\ldots,g_k$ are functions,
then $\ang{\cF_1,\ldots,\cF_j,g_1,\ldots,g_k}:=\ang{\cF_1\cup\cdots \cup \cF_j\cup\{g_1,\ldots,g_k\}}$.
 
Note that if $\cF\sse\cB$ and $g\in\ang{\cF}$, then $\ang{\cF,g}=\ang{\cF}$ \cite[Lemma 2.1]{bulatov_expressibility_2013}.

\begin{dfn}[\cite{bulatov_expressibility_2013}]\label{def:pps-omega}
 A function $f\in\cB_k$ is \emph{pps$_\om$-definable} over $\cF\sse\cB$ if there exists a finite subset $S_f$ of $\cF$ such that, for every $\epsilon>0$, there exists a $k$-ary function $f'$, pps-definable over $S_f$, such that
 \[
  \left\| f'-f \right\|_\infty = \max_{\bx\in\{0,1\}^k} \left|f'(\bx)-f(\bx)\right| < \epsilon.
 \]
\end{dfn}

Note that to be pps$_\om$-definable over some subset of $\cB$, a function must itself be in $\cB$, i.e.\ it must take rational values.
This avoids complications resulting from the limits of some rational sequences being irrational.

\begin{dfn}
 Let $\cF\sse\cB$.
 The set of all functions that are pps$_\om$-definable over the set $\cF\cup\{\EQ\}$ is called the \emph{$\om$-clone} generated by $\cF$; it is denoted $\ang{\cF}_\om$.
\end{dfn}

In \cite{bulatov_expressibility_2013}, $\om$-clones were originally called ``pps$_\om$-definable functional clones''; we instead use the shorter terminology from \cite{bulatov_functional_2017}.

\begin{dfn}[\cite{bulatov_expressibility_2013}]
 A function $f\in\cB$ is \emph{efficiently pps$_\om$-definable} over $\cF$ if there is a finite subset $S_f$ of $\cF$ and a Turing machine $\cM_{f,S_f}$ with the following property: on input $\epsilon>0$, $\cM_{f,S_f}$ computes a pps-formula $\psi$ over $S_f$ such that $f_\psi$ has the same arity as $f$ and $\|f_\psi-f\|_\infty < \epsilon$.
 The running time of $\cM_{f,S_f}$ is at most a polynomial in $\log\epsilon^{-1}$.
\end{dfn}

\begin{dfn}\label{def:ompclone}
 Let $\cF\sse\cB$.
 The set of all functions that are efficiently pps$_\om$-definable over the set $\cF\cup\{\EQ\}$ is called the \emph{$(\om,p)$-clone} generated by $\cF$; it is denoted $\ang{\cF}_{\om,p}$.
\end{dfn}

The $(\om,p)$-clones were originally called ``efficiently pps$_\om$-definable functional clones''; we shorten the terminology here to bring it in line with the term ``$\om$-clone''.
We use the same shorthand as for functional clones, so
$\ang{\cF_1,\ldots,\cF_j,g_1,\ldots,g_k}_\om:=\ang{\cF_1\cup\cdots \cup \cF_j\cup\{g_1,\ldots,g_k\}}_\om$
and 
$\ang{\cF_1,\ldots,\cF_j,g_1,\ldots,g_k}_{\om,p}:=\ang{\cF_1\cup\cdots \cup \cF_j\cup\{g_1,\ldots,g_k\}}_{\om,p}$
 
Note that if $\cF\sse\cB$ and $g\in\ang{\cF}_{\om}$, then $\ang{\cF,g}_{\om}=\ang{\cF}_{\om}$ \cite[Lemma 2.2]{bulatov_expressibility_2013}.
Similarly, if $\cF\sse\cB$ and $g\in\ang{\cF}_{\om,p}$, then $\ang{\cF,g}_{\om,p}=\ang{\cF}_{\om,p}$ \cite[Lemma 2.4]{bulatov_expressibility_2013}.

\begin{obs}\label{obs:pinning}
 If $\upf\in\cB^{<,\mathrm{n}}_1$, then $\delta_1\in\ang{\upf}_{\om,p}$ and if $\downf\in\cB^{>,\mathrm{n}}_1$, then $\delta_0\in\ang{\downf}_{\om,p}$.
 (Indeed, let $\upf\in\cB^{<,\mathrm{n}}_1$, that is $\upf(0)=a,\upf(1)=1$ for some $a<1$. Then $\upf^k(0)=a^k$ and $\upf^k(1)=1$. Therefore $\lim_{k\to \infty}\upf^k(x)=\delta_1$, witnessing that $\delta_1\in\ang{\upf}_{\om,p}$.)
\end{obs}

Recall the definition of $\NEQ$ from Section~\ref{sec:binfn}
and the definition of $\cB_1$ from Section~\ref{s:definitions}.
We next define the functional clone~$\cN$, which arises   in our classification theorems.
Elements of $\cN$ are called \emph{product type functions}.

\begin{dfn}\label{def:cN}
$\cN := \ang{\NEQ,\cB_1}$.
\end{dfn}

Two collections of functions defined earlier in this paper are in fact $\om$-clones: the log-supermodular functions and the functions with nonnegative Fourier transform.

\begin{lem}[{\cite[Lemma 4.2]{bulatov_expressibility_2013}}]\label{lem:LSM_closed_omega}
 If $\cF\subseteq\LSM$, then $\ang{\cF}_{\om}\subseteq\LSM$.
\end{lem}

\begin{lem}[{\cite[Theorem 28]{bulatov_functional_2017}}]\label{lem:Fou_omega}
 $\ang{\cP}_\om=\cP$.
\end{lem}

\subsection{Approximate counting and counting CSPs}
\label{s:counting_CSPs}

A counting constraint satisfaction problem $\NCSP(\cF)$ is parameterised by a finite set of functions $\cF$ over some domain $D$, which we take to be $D=\{0,1\}$.
An instance $\Omega$ of $\NCSP(\cF)$ is specified by a finite set $V$ of variables and a finite set $C$ of constraints.
Each constraint $c=(\bv_c,f_c)$ consists of a tuple $\bv_c$ of $k$ variables for some $k\in\zN$, which may involve repeated variables, and a function $f_c\in\cF$ of arity $k$.
Any assignment $\sigma:V\to\{0,1\}$ of values to the variables is thus associated with a weight $w_\sigma = \prod_{c\in C} f_c(\sigma(\bv_c))$, where $\sigma(\bv_c)$ is computed component-wise.
The \emph{partition function} of instance $\Omega$ is the sum of the weights of all the different assignments:
\begin{equation}\label{eq:bit_flip_partition_function}
 Z(\Omega) = \sum_{\sigma:V\to\{0,1\}} w_\sigma = \sum_{\sigma:V\to\{0,1\}} \prod_{c\in C} f_c(\sigma(\bv_c)).
\end{equation}

We will be interested in the problem of approximating the partition function associated with a counting CSP.
An approximation  scheme for the problem $\NCSP(\cF)$ is an algorithm that takes as input an error parameter $\epsilon$ and an instance $\Omega$, and outputs a value $\tilde{Z}$, which satisfies
\[
 e^{-\epsilon}Z(\Omega) \leq \tilde{Z} \leq e^\epsilon Z(\Omega).
\]
This algorithm is an \emph{FPTAS} or \emph{fully polynomial-time approximation scheme} if its running time is polynomial in the size of the instance and in $\epsilon^{-1}$.

The problem $\NCSP(\cF)$ has a \emph{randomised approximation scheme} if there exists a randomised algorithm which takes as input an instance $\Omega$ and an error parameter $\epsilon>0$, and which outputs a value $\tilde{Z}$, satisfying
\[
 \Pr\left[ e^{-\epsilon}Z(\Omega) \leq \tilde{Z} \leq e^\epsilon Z(\Omega) \right] \geq \frac{3}{4}
\]
for all inputs $\Omega$ and $\epsilon$.
This algorithm is an \emph{FPRAS} or \emph{fully polynomial randomised approximation scheme} if its running time is polynomial in the size of the instance and in $\epsilon^{-1}$ \cite{dyer_relative_2004}.

An \emph{approximation-preserving reduction} (or \emph{AP-reduction}) from a counting problem $C$ to a counting problem $C'$ is an algorithm turning an FPRAS for $C'$ into an FPRAS for $C$.
A rigorous definition of this notion may be found in \cite{dyer_relative_2004}.
If there exists such a reduction, we write $C\leq_{AP}C'$ and say $C$ is \emph{AP-reducible} to $C'$.
If AP-reductions exist in both directions, $C$ and $C'$ are said to be \emph{AP-interreducible}, denoted $C=_{AP} C'$.

\begin{obs}\label{obs:bit_flip_CSP}
 Let $\cF$ be a finite subset of $\cB$, then $\NCSP(\cF) =_{AP} \NCSP(\bar{\cF})$.
\end{obs}

This holds because the value of the partition function \eqref{eq:bit_flip_partition_function} is unchanged if each constraint function in that equation is replaced by its bit-flip.
We use the same shorthand for counting CSPs that we do for clones, so
$\NCSP(\cF_1,\ldots,\cF_j,g_1,\ldots,g_k):=\NCSP(\cF_1\cup\cdots \cup \cF_j\cup\{g_1,\ldots,g_k\})$.

\begin{lem}\label{lem:normalising}
 For all $f\in\cB^{<}_1$ there is an $f'\in\cB^{<,\mathrm{n}}_1$ such that, for all finite $\cF\sse\cB$, the problem $\NCSP(\cF, f) =_{AP} \NCSP(\cF,f')$.
 Similarly, for all $f\in\cB^{>}_1$ there is an $f'\in\cB^{>,\mathrm{n}}_1$ such that, for all finite $\cF\sse\cB$, $\NCSP(\cF, f) =_{AP} \NCSP(\cF,f')$.
\end{lem}
\begin{proof}
 To see the first statement, note that $f\in\cB^{<}_1$ implies $f$ is permissive.
 The function $f'(x)=f(x)/f(1)$ is therefore well-defined, and an element of $\cB^{<,\mathrm{n}}_1$.
 Consider an instance $\Omega$ of $\NCSP(\cF, f)$ with $n$ constraints, $m$ of which use the function $f$.
 Construct an instance $\Omega'$ of $\NCSP(\cF, f')$ by replacing each of those $m$ constraints with a constraint using $f'$.
 Then $Z(\Omega) = (f(1))^m Z(\Omega')$, so since $m=\mathcal{O}(n)$, any FPRAS for $\NCSP(\cF, f')$ can be turned into an FPRAS for $\NCSP(\cF, f)$.
 Similarly, any FPRAS for $\NCSP(\cF, f)$ can be turned into an FPRAS for $\NCSP(\cF, f')$, so $\NCSP(\cF, f) =_{AP} \NCSP(\cF,f')$.
 
 An analogous argument holds if $f\in\cB^{>}_1$.
\end{proof}

The complexity of counting CSPs is closely linked to the theory of $(\om,p)$-clones, as can be seen from the following AP-reduction.

\begin{lem}[{\cite[Lemma 10.1]{bulatov_expressibility_2013}}]\label{lem:clone_csp}
 Suppose that $\cF$ is a finite subset of $\cB$.
 If $g\in\ang{\cF}_{\om,p}$, then
 \[
  \NCSP(\cF,g) \leq_{AP} \NCSP(\cF).
 \]
\end{lem}

We will frequently use the following subset of CSPs called \emph{Holant problems}~\cite{Cai18}.
\begin{dfn}\label{def:Holant}
 A counting CSP in which each variable appears exactly twice is called a \emph{holant problem}.
 Thus, an instance of $\hol(\cF)$ is an instance $\Omega=(V,C)$ of $\NCSP(\cF)$ such that each variable in~$V$ appears exactly twice in constraints in~$C$.
 The objective of the problem $\hol(\cF)$ is to compute $Z(\Omega)$.
\end{dfn}
There is an equivalent view of holant problems that will also be useful.
Let $\EQ_3$ be the ternary equality function, that is, $\EQ_3(x,y,z)=1$ if and only if $x=y=z$ and $\EQ_3(x,y,z)=0$ otherwise.
It turns out (see~\cite[Proposition 1]{cai_Holant_2012} )
that the problem $\NCSP(\cF)$ is equivalent to $\hol(\cF,\EQ_3)$ 
in  following sense.
 For each instance $\Omega$ of $\NCSP(\cF)$ there is an instance 
$\Omega'$ of $\hol(\cF,\EQ_3)$ 
such that $Z(\Omega) = Z(\Omega')$. There is an easy polynomial-time algorithm
that turns $\Omega$ into~$\Omega'$ and vice-versa.

Holant problems are an object of study in their own right and exhibit some properties not found in counting CSPs.
For example, the holant framework allows reductions by \emph{holographic transformations}, which transform all the constraint functions but nevertheless keep the partition function invariant.
These transformations are easiest to define by exploiting the following bijection between functions in $\cB_k$ and vectors in $(\zQ_{\geq 0})^{2^k}$:
\[
 f\in\cB_k \quad\leftrightarrow\quad \mathbf{f}=(f(0\zd 0), f(0\zd 0,1) \zd f(1\zd 1)).
\]
Let $M$ be an invertible $2\times 2$ matrix with nonnegative  rational values. The holographic transformation of $f$ by $M$, denoted $M\circ f$, is the function corresponding to $(M\otimes\ldots\otimes M)\mathbf{f}$, where the tensor product contains $k$ copies of $M$.
For a set $\cF\sse\cB$, we define $M\circ\cF := \{M\circ f : f\in\cF\}$.
We can now state a corollary of Valiant's Holant Theorem \cite{valiant_holographic_2008,cai_Holant_2012}, adapted to our setting.
The result follows from \cite[Proposition~5]{cai_Holant_2012} because constant factors can be absorbed into AP-reductions.

\begin{thm}[Corollary of Valiant's Holant Theorem]\label{thm:valiant_Holant}
Let $\cF$ be a finite subset of $\cB$ and let
$M$ be a $2\times 2$ matrix such that $M M^{T} = c I$
where $I$ is the identity matrix and
$c$ is a rational number.
Then  
 $
  \hol(\cF) =_{AP} \hol(M \circ\cF)$.
\end{thm}

A holant instance $(V,C)$ can be represented as a multigraph 
with vertex set $C$.
The edges of the multigraph correspond to the variables in~$V$.
Informally, there is a self-loop for each variable that appears twice in the same constraint.
If the two appearances of a variable $v$ are in different constraints, say $c$ and $c'$, then there is an edge of the multigraph from $c$ to $c'$ corresponding to $v$. 
This notion is formalised in the following observation and illustrated by example in Figure~\ref{fig:Holant_multigraph}.

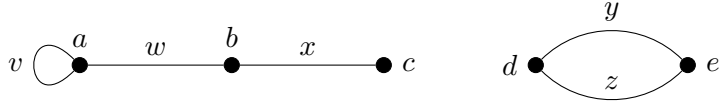
\begin{figure}
 \centering
  \begin{tikzpicture}[scale=2]
  \node[circle, inner sep=2pt, fill, draw, label={above:$a$}] (d0) at (0, 0) {};
  \node[circle, inner sep=2pt, fill, draw, label={above:$b$}] (d1) at (1, 0) {};
  \node[circle, inner sep=2pt, fill, draw, label={right:$c$}] (d2) at (2, 0) {};
  \node[circle, inner sep=2pt, fill, draw, label={left:$d$}] (d3) at (3, 0) {};
  \node[circle, inner sep=2pt, fill, draw, label={right:$e$}] (d4) at (4, 0) {};
  \draw (d0) to[out=135, in=-135, loop] node[left] {$v$} ();
  \draw (d0) to node[above] {$w$} (d1);
  \draw (d1) to node[above] {$x$} (d2);
  \draw (d3) to[out=45, in=135] node[above] {$y$} (d4);
  \draw (d3) to[out=-45, in=-135] node[above] {$z$} (d4);
  \end{tikzpicture}
 \caption{The multigraph representation of a holant instance $\Omega=(V,C)$ with variables $V=\{v,w,x,y,z\}$ and constraints $C =\{a,b,c,d,e\}$ with $a=((v,v,w),f)$, $b=((w,x),g)$, $c=((x),h)$, $d=((y,z),g)$, and $e=((y,z),k)$, constructed as in Observation~\ref{obs:Holant_graph}. Note that $V'=\{a,b,c,d,e\}$, $S=\{\{a,a\}\}$, $T=\{\{a,b\},\{b,c\}, \{d,e\},\{d,e\}\}$, and $E'=S\cup T$. The set of constraint functions used is $\cF=\{f,g,h,k\}$.}
 \label{fig:Holant_multigraph}
\end{figure}

\begin{obs}\label{obs:Holant_graph}
 Any holant instance $\Omega$ with variables $V$ and constraints $C$ corresponds to a multigraph $G'=(V',E')$ 
 where $V'=C$ and $E'$ is defined as follows.
  Let $S$ be the multiset whose elements are of the form $\{c,c\}$ 
 where  $c\in C$. The multiplicity of $\{c,c\}$ in $S$ is
equal to the number of variables that are repeated in the scope $\bv_c$.
 Let $T$ be the multiset whose elements are of the form $\{c,c'\}$ 
 where $c\in C$ and $c'\in C$ but $c\neq c'$.
 The multiplicity of $\{c,c'\}$ in $T$ is 
 equal to the number of variables that are in the scope of $\bv_c$ and the scope of $\bv_{c'}$.
 $E'$ is defined to be the multiset $S\cup T$. 
\end{obs}

For asymmetric constraints it is necessary to specify an enumeration of the edges incident on each vertex in order to be able to recover the original holant instance from the graph.

 Many counting problems defined on graphs have natural expressions in the holant framework
 which are easily understood using this perspective.
For example, the problem of counting the perfect matchings 
of a graph
corresponds to $\hol\left(\left\{\textsc{Exact-One}_k : k\in\zN_{>0} \right\}\right)$ with
\[
 \textsc{Exact-One}_k(\vc{x}{k}) = \begin{cases} 1 &\text{if } \sum_{i=1}^k x_i = 1 \\ 0 & \text{otherwise.} \end{cases}
\]
An edge is in the perfect matching if the corresponding variable is assigned $1$.
The constraint functions ensure that any assignment with non-zero weight includes exactly one edge incident on any vertex, and thus corresponds to a perfect matching of the graph.

The problem of counting the number of satisfying assignments of a Boolean formula in conjunctive normal form is denoted $\SAT$.
This problem, and any problem that $\SAT$ AP-reduces to, cannot have an FPRAS unless $\NP=\RP$ \cite{dyer_relative_2004}.

The problem of counting the number of independent sets in a bipartite graph is denoted $\BIS$.
While $\BIS\leq_{AP}\SAT$, no AP-reduction from $\SAT$ to $\BIS$ is known.
At the same time, $\BIS$ is not known to have an FPRAS \cite{dyer_relative_2004}.

\subsection{Existing results}\label{sec:existing}

There is a trichotomy for the complexity of counting CSPs where all constraints are Boolean relations.
This can straightforwardly be transformed into a result about pure pseudo-Boolean functions as multiplying a constraint function by a nonnegative rational constant does not affect the complexity of the counting problem with respect to AP-reductions.

\IMtwo is the set of relations that can be expressed as conjunctions of the unary relations $\{(0)\}$ and $\{(1)\}$ (which correspond to the functions $\delta_0$ and $\delta_1$) as well as the binary relation $\IMP=\{(0,0),(0,1),(1,1)\}$.

\begin{thm}[{\cite[Theorem~3]{Dyer10:approximation}}]\label{thm:complexity_Boolean}
 Let $\Gm$ be a set of Boolean relations.
 \begin{enumerate}
  \item If every relation in $\Gm$ is affine then, for any finite subset $S$ of $\Gm$, $\NCSP(S)$ is in \FP.
  \item Otherwise, if every relation in $\Gm$ is in \IMtwo, then 
  \begin{itemize}
  \item for any finite subset $S$ of $\Gm$, $\NCSP(S) \leq_{AP} \BIS$, and
  \item there is a finite subset $S$ of $\Gm$ such that $\BIS \leq_{AP} \NCSP(S)$.
  \end{itemize}
  \item Otherwise, there is a finite subset $S$ of $\Gm$ such that $\NCSP(S)=_{AP}\SAT$.
 \end{enumerate}
\end{thm}

We also use the following complexity classification from~\cite{bulatov_expressibility_2013}.
It is called a ``conservative'' classification because of the way that unary functions
(those in the set $S$ below) are considered, even though they may not
belong to $\cF$.
In essence, this theorem is Theorem 10.2 of~\cite{bulatov_expressibility_2013}.
However, we quote the version from \cite[Lemma 7]{chen_complexity_2015},
where the ranges of functions are restricted to rational numbers. 
 \begin{thm}[{\cite{bulatov_expressibility_2013}}]\label{thm:complexity_conservative}
 Suppose $\cF$ is a  finite subset of $\cB$.
 \begin{enumerate}
  \item If $\cF\sse\cN$ then, for any finite subset $S$ of $\cB_1$, 
  $\NCSP(\cF,S)$ is in \FP.   
  \item Otherwise,
   \begin{enumerate}
    \item There is a finite subset $S$ of $\cB_1$ such that $\BIS\leq_{AP}\NCSP(\cF,S)$.
    \item If $\cF\nsubseteq\LSM$ then there is a finite subset $S$ of $\cB_1$ such that $\SAT=_{AP}\NCSP(\cF,S)$.
   \end{enumerate}
 \end{enumerate}
\end{thm}
The statement of the theorem in~\cite{bulatov_expressibility_2013, chen_complexity_2015}
only guarantees an FPRAS when $\cF\sse\cN$
but this is because  \cite{bulatov_expressibility_2013}
was working over (approximable) real numbers.
The proof guarantees an exact algorithm in our setting where elements of the range of functions are rational.

We also require some results about the complexity of symmetric antiferromagnetic two-spin systems
(on simple graphs) with external fields.
A symmetric two-state spin system with parameters $(\beta,\gm,\ld)\in\zQ_{\geq 0}^2\times\zQ_{>0}$ corresponds to the problem $\NCSP(f,g)$, where $g$ is the unary function $g(x) = \ld^{1-x}$, and $f$ is the symmetric binary function
\[
 f(x,y) = \begin{pmatrix}\beta&1\\1&\gm\end{pmatrix}.
\]
The binary constraints are characterised by an undirected simple graph $G=(V,E)$ whose set of vertices $V$ is the set of variables of the instance.
Formally, all constraints must satisfy the following conditions:
\begin{itemize}
 \item For every vertex $v\in V$, there is exactly one constraint of the form $((v),g)$, and there are no other unary constraints.
 \item For every undirected edge $e=\{v,w\}\in E$, there is exactly one constraint of the form $((v,w),f)$ or $((w,v),f)$ and there are no other constraints with scope $(v,w)$ or $(w,v)$.
 \item For any pair $v',w'\in V$ such that $\{v',w'\}\notin E$, there are no constraints with scope $(v',w')$ or $(w',v')$.
\end{itemize}
The spin system is antiferromagnetic if $\beta\gm<1$.
By Observation~\ref{obs:binary_lsm}, this is exactly the same as saying that $f$ is not lsm.
The following results additionally assume $\beta\leq\gm$.

In the following, $\Delta$ denotes the maximum degree of the graph describing the binary constraints in the given instance of $\NCSP(f,g)$.
The expression ``$\Delta=\infty$'' indicates that the result applies to graphs of unbounded degree.

\begin{thm}[{\cite[Theorem 1.2]{Li11:correlation}}]\label{thm:uniqueness_FPTAS}
 For any finite $\Delta\geq 3$ or $\Delta=\infty$, there exists an FPTAS for the partition function of the [symmetric] two-state antiferromagnetic spin system on graphs of maximum degree at most $\Delta$ if for all $d\leq\Delta$ the system parameters
 $(\beta,\gamma,\lambda)$
  lie in the interior of the uniqueness region of the infinite $d$-regular tree.
\end{thm}

See also related work of Sinclair, Srivastava and Thurley \cite{sinclair_approximation_2014}.

\begin{thm}[{\cite{sly_computational_2012,galanis_inapproximability_2016}, as stated in \cite[Theorem 1.3]{Li11:correlation}}]\label{thm:non-uniqueness_hard}
 For any finite $\Delta\geq 3$ or $\Delta=\infty$, unless $\NP=\RP$, there does not exist an FPRAS for the partition function of the [symmetric] two-state antiferromagnetic spin system on graphs of maximum degree at most $\Delta$ if for some $d\leq\Delta$ the system parameters
 $(\beta,\gamma,\lambda)$
  lie in the interior of the non-uniqueness region of the infinite $d$-regular tree.
\end{thm}

These two theorems classify most symmetric antiferromagnetic two-spin system 
with external fields -- the only case that is still open is that of system parameters on the boundary between the uniqueness and non-uniqueness regions.

We will only need a subset of the properties proved in \cite[Lemma 3.1]{Li11:correlation} and therefore state only these.
In the following, ``up-to-$\Delta$ unique'' means that the system parameters lie in the interior of the uniqueness region of the infinite $d$-regular tree for every $d\leq\Delta$, and ``universally unique'' means that the system parameters lie in the interior of the uniqueness region of the infinite $d$-regular tree for every $d$.

\begin{lem}[{\cite[Lemma 3.1]{Li11:correlation}}]\label{lem:uniqueness_conditions}
 Let $(\beta,\gm,\ld)$ be antiferromagnetic.
 \begin{description}
  \item [(2)] If $\gm\leq 1$, then uniqueness does not hold on the infinite $d$-regular tree for all sufficiently large $d$.
  \item [(5)] If $\beta=0$, for any $\Delta$, there exists a critical threshold $\ld_c=\ld_c(\gm,\Delta) = \min_{1<d<\Delta} \frac{\gm^{d+1} d^d}{(d-1)^{d+1}}$ such that $(\beta,\gm,\ld)$ is up-to-$\Delta$ unique if and only if $\ld\in(0,\ld_c)$.
  \item [(8)] If $\beta>0$ and $\gm>1$, there exists an absolute positive constant $\ld_c=\ld_c(\beta,\gm)$ such that $(\beta,\gm,\ld)$ is universally unique if and only if $\ld\in(0,\ld_c)$.
 \end{description}
\end{lem}

In fact, the proof of (2) in the full version of \cite{Li11:correlation} (where this result is part of Lemma~21) shows that, for sufficiently large $d$, the system parameters lie in the interior of the non-uniqueness region.
The proof of (5) and (8) similarly shows that the system parameters lie in the interior of the non-uniqueness region when $\ld > \ld_c$.

\section{Counting CSPs with strictly increasing and strictly decreasing permissive unary functions}
\label{s:ups_and_downs}

We first consider the complexity of counting CSPs where two permissive unary functions are available: one which is strictly increasing and one which is strictly decreasing, i.e.\ problems of the form $\NCSP(\cF,\upf,\downf)$.
The goal is to prove the complexity classification in Theorem~\ref{thm:up-down}.

The proof splits into several cases, given here as individual lemmas, depending on which binary functions are contained in $\Gud$.
The function $\EQ(x,y)$ is contained in any $(\om,p)$-clone, so there always exists some binary function in $\Gud$.

In some lemmas, we will also assume that the permissive unary functions $\upf$ and $\downf$ are normalised.
This can be achieved using AP-reductions, as shown in Lemma~\ref{lem:normalising}.

\subsection{Non-lsm functions with nontrivial binaries}

First, we consider the case where $\cF\nsubseteq\LSM$ and $\Gud$ contains a nontrivial binary function.

\begin{lem}\label{lem:get-nontrivial-binary}
 Let $\cF\sse\cB$, $\upf\in\cB^{<,\mathrm{n}}_1$, and $\downf\in\cB^{>,\mathrm{n}}_1$.
 Suppose $f,g\in\Gud$, where $f$ is non-lsm and $g$ is binary and nontrivial (but may be lsm). 
 Then $\Gud$ contains a binary nontrivial non-lsm function.
\end{lem}

\begin{proof}
By Observation~\ref{obs:pinning}, $\delta_0,\delta_1\in\Gud$.

Since $f$ is non-lsm, there are $\ba,\bb\in\{0,1\}^r$, where $r$ is the arity of 
$f$, such that
\[
  f(\ba)f(\bb)>f(\ba\meet\bb)f(\ba\join\bb).
\]
Suppose $\ba[i]=\bb[i]$ for some $i\in[r]$.
Then let 
\[
  h(x_1\zd x_{i-1}, x_{i+1}\zd x_r) = \sum_{x_i\in\{0,1\}} \delta_{\ba[i]}(x_i) f(x_1\zd x_r),
\]
this function is contained in $\Gud$.
We say $h$ is $f$ with the $i$-th input pinned to $\ba[i]$.
Define $\ba'=(\ba[1]\zd\ba[i-1],\ba[i+1]\zd\ba[r])$, and similarly $\bb'$.
Now,
\[
  h(\ba')h(\bb') = f(\ba)f(\bb)>f(\ba\meet\bb)f(\ba\join\bb) = h(\ba'\meet\bb')h(\ba'\join\bb'),
\]
so $h$ is non-lsm.
Thus we may continue the proof with $h$ in place of $f$.
This process can be repeated until the bit-strings $\ba,\bb$, which witness that $f$ is non-lsm, satisfy $\ba[i]\ne\bb[i]$ for all $i\in[r]$.

Without loss of generality, we may furthermore assume that $\ba=(0^s,1^t)$ and $\bb=(1^s,0^t)$: otherwise permute the arguments, which does not affect the non-lsm property.
Now, identify the first $s$ and the last $t$ variables of $f$ to obtain a binary function $f'$ satisfying
\[
  f'(\ba'')f'(\bb'') > f'(\ba''\meet\bb'')f'(\ba''\join\bb''),
\]
where $\ba''=(0,1), \bb''=(1,0), \ba''\meet\bb''=(0,0)$ and $\ba''\join\bb''=(1,1)$.

If $f'(0,0)\ne0$ or $f'(1,1)\ne0$, then $f'$ is a binary nontrivial non-lsm function and we are done.
Otherwise $f'$ has the form $f''(x)\NEQ(x,y)$, say, $f''(0)=c,f''(1)=d$ with $c,d>0$. 

By the assumptions of the lemma, there is a binary nontrivial function $g\in\Gud$. 
Suppose
\[
  g(x,y)=\binar\al\beta\gm\dl.
\]
If $g$ is non-lsm, we are done.
Otherwise, we have $\al\dl>\beta\gm$.
This inequality is strict because $g$ is not log-modular.
Then let $g'(x,y)$ be given by
\[
  g'(x,y) = \sum_{z\in\{0,1\}} g(x,z)f'(z,y) = \binar{\beta d}{\al c}{\dl d}{\gm c}.
\]
Thus, 
\[
  g'(0,0) g'(1,1) = \beta \gm cd < \al \dl cd = g'(0,1) g'(1,0),
\]
that is, $g'$ is non-lsm.
It is also nontrivial because $g$ being nontrivial implies that at least one of $\beta,\gm$ is strictly positive.
This completes the proof.
\end{proof}

\subsection{Log-supermodular functions with nontrivial binaries}

Suppose now that $\cF\sse\LSM$ and $\Gud$ contains a nontrivial binary function.
The first property is equivalent to $\Gud\sse\LSM$ by Lemma~\ref{lem:LSM_closed_omega} and the  fact that $\cB_1\sse\LSM$.

\begin{lem}\label{lem:lsm-nontrivial}
 Let $\cF$ be a finite subset of $\cB$ and let $\upf\in\cB^{<}_1$, $\downf\in\cB^{>}_1$.
 If $\Gud\sse\LSM$ and there is a nontrivial binary function $g\in\Gud$, then  $\NCSP(\cF,\upf,\downf)$ is \BIS-hard. 
\end{lem}

\begin{proof}
 The nontrivial binary function $g$ may be assumed symmetric by replacing it with the function 
 $\sum_{z\in\{0,1\}} g(x,z) g(y,z)$ 
(which is contained in $\ang{g}$) 
 if necessary.
 By Lemma~\ref{lem:symmetric}\,\eqref{p:symmetric_lsm}, the  proposed replacement function is also nontrivial and lsm.
 A nontrivial symmetric lsm function can be written as
 \[
  g(x,y)=c\binar\beta11\gm,
 \]
 where $\beta\gamma>1$ (cf.\ Observation~\ref{obs:binary_lsm}) and $c>0$.
 Let $g'(x,y) := \frac{1}{c}g(x,y)$ and let $h(x) := \left(\frac{\downf(x)}{\downf(1)}\right)^k$ for some 
 positive integer $k$ that we will determine below.
 Note that $h(1)=1$ and let $\mu:=h(0)$, which is strictly greater than $1$.
 Constant factors do not affect the complexity of CSPs, so $\NCSP(g',h) \le_{AP} \NCSP(g,\downf)$.
 We now distinguish cases according to the relative size of $\beta$ and $\gamma$.
 
 \textbf{Case 1}. Suppose $\beta<\gamma$.
 Then by \cite[Theorem 2]{liu_complexity_2014}, the problem $\NCSP(g',h)$ is \BIS-hard if $\mu$ is sufficiently large. 
 But $\mu=\left(\frac{\downf(0)}{\downf(1)}\right)^k$, which can be made arbitrarily large by choosing $k$ large enough.
 Hence $\NCSP(g,\downf)$ is \BIS-hard, and thus by Lemma~\ref{lem:clone_csp}, $\NCSP(\cF,\upf,\downf)$ is \BIS-hard.
 
 \textbf{Case 2}. Suppose $\beta>\gamma$.
 By Observation~\ref{obs:bit_flip_CSP}, $\NCSP(\cF,\upf,\downf) =_{AP} \NCSP(\overline{\cF},\overline{\upf},\overline{\downf})$.
 But $\overline{\upf}\in\cB^{>}_1$ and $\overline{\downf}\in\cB^{<}_1$, and $\overline{g}(0,0)<\overline{g}(1,1)$.
 Thus we can apply the argument of Case~1 to the bit-flipped functions to find that $\NCSP(\cF,\upf,\downf)$ is again \BIS-hard.

 \textbf{Case 3}. Suppose $\beta=\gm$, i.e.\ $g'$ is a ferromagnetic Ising function.
 Historically, this case was considered first, in a paper by Goldberg \& Jerrum \cite{goldberg_complexity_2007}.
 However, the main result of~\cite{goldberg_complexity_2007} (Theorem 1.1) is stated in a setting where
 the input has ``local fields'', which means that different unary functions are available for different
 variables. 
 While~\cite{goldberg_complexity_2007} does contain a reduction from a problem with restricted local fields  to the general problem, the proof of this result is only given for an explicit choice of  fields.
 This proof is widely known to generalise, as e.g.\ noted in \cite{liu_complexity_2014}, but rather than writing out the general proof here, it will be shorter (if ahistorical) to note the following:
 Let $g''(x,y) = \sum_{z\in\{0,1\}} g(x,z) g(y,z) \upf(z)$.
 By Lemma~\ref{lem:symmetric}\,\eqref{p:Ising}, this function is nontrivial, symmetric, lsm, and non-Ising.
 Thus the problem reduces to one of the previous cases.
 
 Therefore, $\NCSP(\cF,\upf,\downf)$ is \BIS-hard whenever $\Gud\sse\LSM$ and  there exists a nontrivial binary function in $\Gud$.
\end{proof}

\subsection{All binary functions are trivial}

Having considered two cases with nontrivial binaries, we now look at $(\om,p)$-clones that do not contain any nontrivial binary functions.
In the following, $\oplus_3(x,y,z)$ is the ternary indicator function for inputs of even parity, i.e.
 \[
  \oplus_3(x,y,z) = \begin{cases} 1 &\text{if } x+y+z \text{ is even,} \\ 0 &\text{otherwise.} \end{cases}
 \]
We will also use the following two lemmas.
\begin{lem}[\cite{topkis_minimizing_1978}, as stated in {\cite[Lemma~5.1]{bulatov_expressibility_2013}}]\label{lem:2-pinnings_log-modular}
 A permissive function $f\in\cB$ is log-modular if and only if every 2-pinning is log-modular.
\end{lem}

\begin{lem}\label{lem:log-modular-product}
 A permissive function $f\in\cB_n$ is log-modular if and only if it is a product of permissive unary functions, i.e.\ $f(\vc{x}{n}) = \prod_{i=1}^n u_i(x_i)$, where $u_i\in\cB_1$ is permissive for all $i\in [n]$.
\end{lem}
\begin{proof}
 Any product of permissive unary functions must be permissive and log-modular, i.e.\ the ``if'' direction is immediate.
 It remains to prove that if $f$ is permissive and log-modular, then it must be a product of permissive unary functions.
 For $n=1$ that result is trivial and for $n=2$ it follows straightforwardly from the definition of log-modularity.
 
 Now assume the desired result holds for some $n\geq 2$ and consider a permissive log-modular function $f\in\cB_{n+1}$.
 Define $f_a(\vc{x}{n}):=f(\vc{x}{n},a)$ for $a\in\{0,1\}$, then
 \[
  f(\vc{x}{n+1}) = f_0(\vc{x}{n})\dl_0(x_{n+1}) + f_1(\vc{x}{n})\dl_1(x_{n+1}).
 \]
 Since $f$ is permissive and log-modular, $f_0$ and $f_1$ must also be permissive and log-modular.
 Hence by the inductive assumption, there exist permissive unary functions $\vc{u}{n}\in\cB_1$ such that $f_0(\vc{x}{n}) = \prod_{i=1}^n u_i(x_i)$, and there exist permissive unary functions $\vc{v}{n}\in\cB_1$ such that $f_1(\vc{x}{n}) = \prod_{i=1}^n v_i(x_i)$.
 Thus we can write
 \[
  f(\vc{x}{n+1}) = \left(\prod_{i=1}^n u_i(x_i)\right)\dl_0(x_{n+1}) + \left(\prod_{i=1}^n v_i(x_i)\right)\dl_1(x_{n+1}).
 \]
 Let $k\in [n]$ be arbitrary and consider some 2-pinning of $f$ that leaves the variables $k$ and $n+1$ untouched, pinning each variable $i\in[n]\setminus\{k\}$ to the value $b_i$.
 This yields the function
 \[
  g_k(x_k,x_{n+1}) := \left(\prod_{i\in[n]\setminus\{k\}} u_i(b_i)\right)u_k(x_k)\dl_0(x_{n+1}) + \left(\prod_{i\in[n]\setminus\{k\}} v_i(b_i)\right)v_k(x_k)\dl_1(x_{n+1}).
 \]
 By Lemma~\ref{lem:2-pinnings_log-modular}, $g_k$ is log-modular, i.e.\
 \[
  0 = g_k(0,0)g_k(1,1)-g_k(0,1)g_k(1,0)
  = \left(\prod_{i\in[n]\setminus\{k\}} u_i(b_i)v_i(b_i)\right) \left( u_k(0)v_k(1) - v_k(0)u_k(1) \right).
 \]
 Since all the unary functions $u_i$ and $v_i$ are permissive and take non-negative rational values, this implies that there exists $c_k\in\zQ_{>0}$ such that $v_k=c_k\cdot u_k$.
 But $k$ was arbitrary, therefore
 \[
  f(\vc{x}{n+1}) = \left(\prod_{i=1}^n u_i(x_i)\right) \left( \dl_0(x_{n+1}) + \dl_1(x_{n+1})\prod_{i=1}^n c_i\right) = \prod_{i=1}^{n+1} u_i(x_i),
 \]
 where $u_{n+1}(0)=1$ and $u_{n+1}(1) = \prod_{i=1}^n c_i$, so $u_{n+1}$ is permissive and in $\cB_1$.
 Hence $f$ has the desired form, completing the proof.
\end{proof}

\begin{lem}\label{lem:trivial_parity}
 Let $\cF\sse\cB$, $\upf\in\cB^{<,\mathrm{n}}_1$, and $\downf\in\cB^{>,\mathrm{n}}_1$.
 Suppose $\Gud$ does not contain any nontrivial binary functions and $\cF\nsubseteq\cN$.
 Then there exists a ternary function $g\in\Gud$ satisfying $g(x,y,z)=c\cdot\oplus_3(x,y,z)$, where $c>0$ is a constant.
\end{lem}

\begin{proof}
 The proof of the lemma has two parts.
 First, we show that if $\Gud$ does not contain any nontrivial binary functions, then every function in $\Gud$ has affine support.
 Next, we show that if every function in $\Gud$ has affine support and $\cF\nsubseteq\cN$, then there exists $g\in\Gud$ which satisfies $g(x,y,z)=c\cdot\oplus_3(x,y,z)$ for some non-zero constant $c$.

 To prove that every function in $\Gud$ has affine support, we assume the opposite and show this leads to a contradiction.
 In particular, assume there exists a function $f\in\Gud$ which does not have affine support.
 All unary functions have affine support, so $f$ cannot be unary.
 Furthermore, all binary functions in $\Gud$ are trivial and a trivial binary function is either log-modular or it has a support on exactly 2 inputs.
 Recall from Section~\ref{sec:binfn} that any log-modular binary function $b(x,y)$ can be written as $u(x)v(y)$ for some $u,v\in\cB_1$; hence it has support on 1, 2 or 4 inputs.
 Now any subset of $\{0,1\}^2$ of size 1, 2 or 4 is an affine relation -- in other words all trivial binary functions have affine support.
 Therefore, the function $f$ cannot be binary, so it has arity at least 3.
 
 By Observation~\ref{obs:pinning}, the pinning functions $\delta_0$ and $\delta_1$ are contained in $\Gud$.
 Now, the proof of Lemma 11 in \cite{dyer_complexity_2009} shows that, if $f$ is a nonnegative function with $\ari(f)>2$ and $f$ does not have affine support, then $\ang{f,\delta_0,\delta_1}$ (and thus $\Gud$) contains a function of arity 2 which does not have affine support.
 This contradicts the assumption that all binary functions in $\Gud$ are trivial.
 Therefore, every function in $\Gud$ must have affine support.
 
 To prove that $\Gud$ contains a scaled parity function, suppose there exists a function $h\in\Gud\setminus\cN$.
 All unary functions are contained in $\cN$ since they are generators of the functional clone $\cN$.
 Furthermore, as noted in Section~\ref{sec:binfn}, any binary log-modular function can be written as $u(x)u'(y)$, where $u,u'\in\cB_1$ are appropriate unary functions.
 Thus, all trivial binary functions are contained in $\cN$.
 Hence, $h$ has arity at least 3.
 It is not the all-zero function, as that is also contained in $\cN$.
 By the first part of this proof, $h$ has affine support.
 Since all binary functions in $\Gud$ are trivial, all 2-pinnings of $h$ (as defined in Section~\ref{s:definitions}) must be trivial.
 
 We now consider the relations underlying some of the functions in $\Gud$. Recall from Section~\ref{s:definitions} that we denote by $R_f$ the relation underlying a function $f$.
 According to Lemma~3.1 of \cite{bulatov_expressibility_2013}, for any set of nonnegative functions $\cG$, we have
 \begin{equation}\label{eq:relations_functional_clone}
   \ang{\{R_f\mid f\in\cG\}}_R = \{R_f\mid f\in\ang{\cG}\}.
 \end{equation}
 Note that the above result is about functional clones, not $(\om,p)$-clones, and allowing limits may change the set of underlying relations, as can be seen for example in Observation~\ref{obs:pinning}.
 Yet it is straightforward to see that
 \[
  \ang{\{R_f\mid f\in\cG\}}_R = \{R_f\mid f\in\ang{\cG}\} \sse \{R_f\mid f\in\ang{\cG}_{\om,p}\}.
 \]
 Furthermore, we already know that $\{R_f\mid f\in\Gud\}$ contains only affine relations.
 These are the only facts we use below.
 
 Let $R_h$ be the relation underlying $h$, which is affine and non-empty.
 It can then be seen from Table~2 and the proof of Proposition~3 of \cite{creignou_plain_basis} that the relational clone $\ang{R_h,\delta_0,\delta_1}_R$ must be either $\mathrm{IR}_2=\ang{\EQ,\delta_0,\delta_1}_R$ or $\mathrm{ID}_1=\ang{\EQ, \NEQ,\delta_0,\delta_1}_R$ or $\mathrm{IL}_2$, the relational clone containing all affine relations.
 For an example of the argument for this, see the proof of Theorem~9.1 in \cite{bulatov_expressibility_2013}.
 The argument in the next  few paragraphs is also (a simplified version of) an argument in that proof.
 
 Suppose that $\ang{R_h,\delta_0,\delta_1}_R\sse\mathrm{ID}_1$ (i.e.\ $\ang{R_h,\delta_0,\delta_1}_R=\mathrm{IR}_2$ or $\ang{R_h,\delta_0,\delta_1}_R=\mathrm{ID}_1$), then in particular $R_h\in\mathrm{ID}_1$.
 The relation $R_h$ is also non-empty because $h\notin\cN$.
 Since $R_h$ is affine, its elements are solutions to a set of linear equations (cf.\ Section~\ref{s:definitions}).
 We can thus find a partition of the arguments of $h$ into a set of \emph{free} variables and a set of \emph{dependent} variables such that for any assignment of values to the free variables there exists a unique assignment of values to the dependent variables for which the resulting tuple of bit values is in $R_h$.
 Let $n:=\ari(h)$ and assume without loss of generality that $\vc{x}{k}$ are the free variables and $x_{k+1}\zd x_n$ are the dependent variables (if necessary, permute and rename variables).
 Define $h'(\vc{x}{k}) := \sum_{x_{k+1}\zd x_n\in\{0,1\}} h(\vc{x}{n})$, then $h'\in\Gud$ and $h'$ is permissive.
 
 As $h'\in\Gud$, all its 2-pinnings are binary functions in $\Gud$ and are thus trivial by assumption.
 Since $h'$ is permissive, by Observation~\ref{obs:permissive_binary_trivial} its 2-pinnings can be trivial only if they are log-modular.
 Therefore, by Lemma~\ref{lem:2-pinnings_log-modular}, $h'$ is log-modular.
 Furthermore, by Lemma~\ref{lem:log-modular-product}, $h'$ is a product of unary functions, so $h'\in\cN$.
 
 Now, in going from $h$ to $h'$, we summed out the dependent variables.
 Thus, for any fixed $\vc{x}{k}$, there is only a single assignment of $x_{k+1}\zd x_n$ such that $h(\vc{x}{n})$ is non-zero.
 Therefore, $h(\vc{x}{n}) = \chi_h(\vc{x}{n}) h'(\vc{x}{k})$, where $\chi_h$ is the indicator function for the relation $R_h$ -- i.e.\ $\chi_h(\vc{x}{n})=1$ if $(\vc{x}{n})\in R_h$ and $\chi_h(\vc{x}{n})=0$ otherwise.
 We assumed $R_h\in\mathrm{ID}_1=\ang{\EQ,\NEQ,\delta_0,\delta_1}_R$, therefore $\chi_h\in\cN$ and thus $h\in\ang{\chi_h,h'}\sse\cN$. 
 But this contradicts the assumption that $h\in\Gud\setminus\cN$.
 So instead we must have $\ang{R_h,\delta_0,\delta_1}_R\nsubseteq\mathrm{ID}_1$.
 The only way for this to happen is if $\ang{R_h,\delta_0,\delta_1}_R = \mathrm{IL}_2$ and thus $\mathrm{IL}_2\sse\{R_f\mid F\in\Gud\}$.
 
 Now, the relation corresponding to $\oplus_3$ is
 \[
  R_{\oplus_3} = \{(0,0,0),(0,1,1),(1,0,1),(1,1,0)\}.
 \]
 This is affine and therefore contained in $\mathrm{IL}_2=\ang{R_h,\delta_0,\delta_1}_R$.
 Hence, by letting $\cG$ equal $\Gud$ in \eqref{eq:relations_functional_clone}, there must be a ternary function $g\in\Gud$ which has $R_{\oplus_3}$ as its underlying relation.
 
 Let $a=g(0,0,0)$, $b=g(0,1,1)$, $c=g(1,0,1)$, and $d=g(1,1,0)$ be the non-zero values of $g$.
 Now consider the binary function $g'(x,y)=\sum_z g(x,y,z)$, which takes the form
 \[
  g'(x,y) = \begin{pmatrix}a&b\\c&d\end{pmatrix}.
 \]
 By construction, $g'\in\Gud$, so it must be trivial.
 As $a,b,c,d>0$, $g'$ must be log-modular, i.e.\ $ad=bc$.
 Similarly, by summing out $y$ or $x$ and applying the log-modularity condition to the resulting binary function, we deduce $ac=bd$ and $ab=cd$.
 Multiply together the first two of these equations to get $a^2cd=b^2cd$, which implies $a=b$ since all four values are positive.
 By symmetry, we find that in fact $a=b=c=d$, i.e. $g(x,y,z)=c\cdot\oplus_3(x,y,z)$.
\end{proof}

\begin{lem}\label{lem:parity_hard}
 Let $\cF$ be a finite subset of $\cB$ and let $\upf\in\cB^{<}_1$, $\downf\in\cB^{>}_1$.
 Suppose $\Gud$ contains a function $g(x,y,z)=c\cdot\oplus_3(x,y,z)$ for some constant $c>0$.
 Then the problem $\NCSP(\cF,\upf,\downf)$ does not have an FPRAS unless $\NP=\RP$.
\end{lem}

\begin{proof}
 We show the desired result by reduction from the problem of approximating weight enumerators of linear codes.
 The definition and reduction follow \cite[p.~1977 and Lemma~13]{dyer_complexity_2009}, with some modifications because that paper was concerned with hardness of exact evaluation.
 
 A \emph{linear code} is specified by a binary generating matrix $A$.
 Let $\Upsilon$ be the linear subspace generated by the rows of $A$ over $\operatorname{GF}(2)$, then any vector in $\Upsilon$ is a code word.
 The \emph{weight enumerator} of the code specified by $A$ with weight parameter $\ld\in\zQ$ is given by $W_A(\ld) := \sum_{\bv\in\Upsilon} \ld^{\abs{\bv}}$, where $\abs{\bv}$ is the Hamming weight of $\bv$.
 The computational problem $\operatorname{WE}(\ld)$ takes as input a matrix $A$ and outputs $W_A(\ld)$.
 
 The linear space $\Upsilon$ can be specified by a pure affine function $h_A$ taking values in $\{0,1\}$, with the arity of $h_A$ being equal to the number of columns of $A$.
 The set of $\{0,1\}$-valued pure affine functions is in bijection with the set of affine relations $\operatorname{IL}_2$.
 Now, 
 by adding new variables and breaking linear equations into pieces, we have $\ang{\oplus_3,\dl_0,\dl_1}_R=\mathrm{IL}_2$.
 Thus, in particular, $h_A\in\ang{\oplus_3,\dl_0,\dl_1}$ for any $A$.
 Define $u_\ld(x)=\ld^x$ for $x\in\{0,1\}$, then $W_A(\ld) = \sum_{\vc{x}{n}\in\{0,1\}} h_A(\vc{x}{n})\prod_{i=1}^n u_\ld(x_i)$.
 Therefore, by Lemma~\ref{lem:clone_csp}, 
 \[
  \operatorname{WE}(\ld) \leq_{AP} \NCSP(\oplus_3,u_\ld,\dl_0,\dl_1).
 \]
 
 Suppose that $\ld = \frac{\upf(1)}{\upf(0)}$, then $u_\ld\in\cB_1$ and $\upf(0)\cdot u_\ld = \upf$.
 We have
 \[
  \NCSP(\oplus_3,u_\ld,\dl_0,\dl_1) \leq_{AP} \NCSP(\cF,\upf,\downf,\oplus_3,u_\ld,\dl_0,\dl_1) \leq_{AP} \NCSP(\cF,\upf,\downf)
 \]
 where the first reduction is because adding more constraint functions cannot make the problem easier and the second reduction is by repeated applications of Lemma~\ref{lem:clone_csp}, since all of $c\cdot\oplus_3$, $\upf(0)\cdot u_\ld,\dl_0,\dl_1$ are in $\Gud$ and constant factors can be absorbed into AP-reductions.
 Hence, combining the different reductions, we find $\operatorname{WE}\left(\frac{\upf(1)}{\upf(0)}\right) \leq_{AP} \NCSP(\cF,\upf,\downf)$.
 
 Now by \cite[Corollary~7]{goldberg_approximating_2013}, the problem of approximating the weight enumerator of a linear code with weight parameter $\ld>1$ does not have an FPRAS unless $\NP=\RP$.
 But $\ld=\frac{\upf(1)}{\upf(0)}>1$ by the definition of $\upf$, hence the desired result follows.
\end{proof}

\subsection{Putting the pieces together}

Recall Theorem~\ref{thm:up-down}, which we can now prove.
\begin{repthm}{thm:up-down}
 \stateupdown{}{}{}{}
\end{repthm}

\begin{proof}  By Lemma~\ref{lem:normalising}, there exist functions $\upf'\in\cB^{<,\mathrm{n}}_1$ and $\downf'\in\cB^{>,\mathrm{n}}_1$ such that
 \[
  \NCSP(S,\upf,\downf) =_{AP} \NCSP(S,\upf',\downf')
 \]
 for any finite $S\sse\cF$.
 By replacing the problem on the left-hand side with the one on the right-hand side, we may therefore assume in the following that the permissive unary functions $\upf$ and $\downf$ are normalised.

 Property~\ref{p:cN_FP} follows from Theorem~\ref{thm:complexity_conservative}.
 
 For Property~\ref{p:lsm_BIS-hard}, suppose $\cF\subseteq\LSM$ and $\cF\nsubseteq\cN$.
 If all binary functions in $\Gud$ are trivial, then by Lemma~\ref{lem:trivial_parity} there is a ternary function $g\in\Gud$ satisfying $g(x,y,z)=c\cdot\oplus_3(x,y,z)$ where $c>0$.
 Note that
 \[
  g(0,1,1) g(1,0,1) = c^2 > 0 = g(0,0,1) g(1,1,1)
 \]
 i.e. $g$ is not lsm.
 But this is a contradiction as the set of lsm functions is closed under taking $\omega$-clones by Lemma~\ref{lem:LSM_closed_omega}.
 Hence we may assume that there is a nontrivial binary function~$f$  in $\Gud$.
 By the definition of $(\omega,p)$-clone (Definition~\ref{def:ompclone}) there is a finite subset $S_f$ of $\cF$ such that $f \in \ang{S_f,\upf,\downf}_{\om,p}$.
 Then by Lemma~\ref{lem:lsm-nontrivial}, $\NCSP(S_f,\upf,\downf)$ is \BIS-hard.   

 Property~\ref{p:lsm_BIS-easy} follows from Part~3 of \cite[Theorem 6]{chen_complexity_2015}, noting that the property of ``weak log-supermodularity'' used there encompasses all binary log-supermodular functions.
 
 Finally, suppose $\cF$ is not a subset of $\cN$, nor is it a subset of $\LSM$.
 Then, if $\Gud$ contains a nontrivial binary function, by Lemma~\ref{lem:get-nontrivial-binary} it contains a nontrivial binary non-lsm function $g$.
 It then contains the function $g(x,y)g(y,x)$, which is symmetric, nontrivial and non-lsm by Lemma~\ref{lem:symmetric}\,\eqref{p:symmetric_non-lsm}.
 This function can be written as
 \[
  g'(x,y)= d \binar\beta11\gm
 \]
 for some $d>0$ and $\beta,\gm\geq 0$.
 By the definition of $(\omega,p)$-clone, there is a finite subset $S$ of $\cF$
 such that $g' \in \ang{S,\upf,\downf}_{\om,p}$.
 We distinguish two subcases.
 
 \textbf{Case 1}. Suppose $\beta\leq \gm$.
 As $g'$ is non-lsm, by Observation~\ref{obs:binary_lsm}, we have $\beta\gm<1$: i.e., the function corresponds to an antiferromagnetic two-spin model.
 Let $f(x,y)=d^{-1}g'(x,y)$ and let $h(x)=\left(\frac{\downf(x)}{\downf(1)}\right)^{k}$ for some 
 sufficiently large positive integer~$k$ (to be determined below).
 Then $h(1)=1$; set $\ld:=h(0)$.
 Since constant factors can be absorbed into AP-reductions, $\NCSP(f,h) =_{AP}\NCSP(g',\downf)$.
 Now $\NCSP(f,h)$ corresponds to the spin system $(\beta,\gamma,\ld)$, cf.\ Section~\ref{sec:existing}.
 We have $0\leq\beta\leq\gamma$ with $\beta\gamma<1$ and $\lambda$ can be made arbitrarily large by choosing $k$ large enough.
 \begin{itemize}
  \item If $\gamma\leq 1$, then Lemma~\ref{lem:uniqueness_conditions}(2) says that for sufficiently large $d$, uniqueness does not hold on the infinite $d$-regular tree.
  \item Otherwise, if $\beta=0$, then by Lemma~\ref{lem:uniqueness_conditions}(5), $(\beta,\gamma,\ld)$ is not up-to-$\Delta$ unique if $\ld$ is large enough.
  \item Otherwise we must have $\beta>0$ and $\gamma>1$. In that case, by Lemma~\ref{lem:uniqueness_conditions}(8), for large enough $\ld$, universal uniqueness does not hold.
 \end{itemize}
 In each of these cases, the proof of Lemma~\ref{lem:uniqueness_conditions} in the full version of \cite{Li11:correlation} shows that $(\beta,\gamma,\ld)$ actually lies in the interior of the non-uniqueness region.
 Now, if $(\beta,\gamma,\ld)$ is in the interior of the non-uniqueness region,
 Theorem~\ref{thm:non-uniqueness_hard}
  (due to \cite{sly_computational_2012,galanis_inapproximability_2016}) 
  shows that 
   $\NCSP(f,h)$ does not have an FPRAS unless $\NP=\RP$. 
 Hence, $\NCSP(g',\downf)$ does not have an FPRAS unless $\NP=\RP$.
 
 \textbf{Case 2}. Suppose $\beta>\gm$.
 By 
 Observation~\ref{obs:bit_flip_CSP}, $\NCSP(g',\upf)=_{AP}\NCSP(\overline{g'},\overline{\upf})$.
Now
 \[
  \overline{g'}(x,y) = d \binar\gm11\beta,
 \]
 so $\overline{g'}(0,0)<\overline{g'}(1,1)$.
 Furthermore, $\overline{\upf}\in\cB^{>}_1$.
 Hence, by the argument of Case~1, $\NCSP(\overline{g'},\overline{\upf})$ does not have an FPRAS unless $\NP=\RP$ and so $\NCSP(g',\upf)$ does not have an FPRAS unless that condition is satisfied.
 
 If, instead, there are no nontrivial binaries in $\Gud$, we get the property of not having an FPRAS unless $\NP=\RP$ from Lemmas~\ref{lem:trivial_parity} and \ref{lem:parity_hard}.
 This completes the proof of Property~\ref{p:non-lsm_hard} and thus of the theorem.
\end{proof}

\section{Two-spin}
\label{s:two-spin}

In this section, we partially classify the complexity of the problem $\NCSP(f)$, where $f\in\cB_2$.
This is similar to the two-spin problem widely studied in  statistical physics and computer science: the 
relevant quantity   is the partition function arising from the pairwise interactions of neighbouring spins.
Progress so far has mainly been restricted to the symmetric case, but we consider more general interaction matrices.
From Theorem~\ref{thm:complexity_conservative}, we know that $\NCSP(f)$ is in $\FP$ if $f\in\cN$.
We furthermore give a complete classification of the complexity of $\NCSP(f)$ if $f$ is lsm, as well as determining the complexity if $f$ is non-lsm and both $f$ and $\bar{f}$ are non-monotone.

Recall from  Section~\ref{sec:binfn} that a binary function is \emph{trivial} if it is log-modular or if $f(x,y)=g(x)\EQ(x,y)$ or $f(x,y)=g(x)\NEQ(x,y)$ for some unary function $g\in\cB_1$.
Recall also that $\wh{f}$ is the Fourier transform of $f$ (cf.\ Section~\ref{s:Fourier}).
We often write $\wh{f}_{xy}$ to denote $\wh{f}(x,y)$.

The proof of Theorem~\ref{thm:two-spin} is split into various lemmas that are stated and proved individually before being assembled into the proof of the theorem in Section~\ref{s:two-spin-proof}.
Throughout, we write
\[
 f(x,y) = \begin{pmatrix} a & b \\ c & d \end{pmatrix}.
\]

\subsection{Functions whose middle Fourier coefficients have opposite signs}

First, we show that we can realise a strictly increasing and a strictly decreasing permissive unary function if the middle Fourier coefficients of $f$, i.e.\ $\wh{f}_{01}$ and $\wh{f}_{10}$, have opposite signs. 
Availability of these two unary functions reduces the problem to the case considered in Section~\ref{s:ups_and_downs}.

\begin{lem}\label{lem:make_up_down}
 Let $f\in\cB_2$ be a nontrivial binary function and suppose $\wh{f}_{01}\wh{f}_{10}<0$.
 Then both $\ang{f}\cap\cB^{<}_1$ and $\ang{f}\cap\cB^{>}_1$ are non-empty.
\end{lem}

\begin{proof}
 The condition $\wh{f}_{01}\wh{f}_{10}<0$ implies that the Fourier coefficients have opposite signs.
 Without loss of generality, we may assume that $\wh{f}_{01}>0$ and $\wh{f}_{10}<0$, i.e. (cf.\ Section~\ref{s:Fourier})
 \begin{align*}
  a+c &> b+d \\
  a+b &< c+d.
 \end{align*}
 Otherwise, replace $f(x,y)$ by $f''(x,y):=f(y,x)$.
 Then, by Observation~\ref{obs:Fourier_swap}, $\wh{f''}_{01}=\wh{f}_{10}>0$ and $\wh{f''}_{10}=\wh{f}_{01}<0$ so the replacement function satisfies the desired property.
 
 Let $\upf(x):=\sum_y f(x,y)$ and $\downf(y):=\sum_x f(x,y)$; then $\upf,\downf\in\ang{f}$.
 These functions satisfy:
 \begin{align*}
  \downf(0)=a+c &> b+d=\downf(1) \\
  \upf(0)=a+b &< c+d=\upf(1).
 \end{align*}
 Furthermore, since $f$ is nonnegative and nontrivial, $\upf$ and $\downf$ are permissive.
 Hence, $\upf\in\ang{f}\cap\cB^{<}_1$ and $\downf\in\ang{f}\cap\cB^{>}_1$, as desired.
\end{proof}

Thus, if $f$ satisfies $\wh{f}_{01}\wh{f}_{10}<0$, we can use the complexity classification from Theorem~\ref{thm:up-down}.

\subsection{Log-supermodular functions}

The case of an lsm function whose middle Fourier coefficients have opposite signs is included in Lemma~\ref{lem:make_up_down}.
Thus, it only remains to consider the case of an lsm function where the two middle Fourier coefficients have the same sign (or at least one of them is zero).

\begin{lem}\label{lem:Fourier_non-zero}
 Suppose $f$ is a nontrivial binary lsm function with $\wh{f}_{01}\wh{f}_{10}\geq 0$.
 Then $\wh{f}_{11}\geq 0$. 
\end{lem}

\begin{proof}
 By Observation~\ref{obs:binary_lsm}, the lsm condition is $ad>bc$, where the inequality is strict as $f$ is nontrivial.
 We distinguish cases according to whether the Fourier coefficients are both zero, both nonpositive or both nonnegative.
 
 \textbf{Case 1}. $\wh{f}_{01} = \wh{f}_{10} = 0$, i.e.
 \begin{align*}
  a+c &= b+d, \\
  a+b &= c+d.
 \end{align*}
 This implies $a=d$ and $b=c$, i.e. $f$ is an Ising function.
 Now,
 \[
  b+c = 2\sqrt{bc} < 2\sqrt{ad} = a+d,
 \]
 where the first step uses $b=c$, the second step the lsm property, and the last step uses $a=d$.
 But $a+d > b+c$ is equivalent to $\wh{f}_{11}>0$, the desired result.

 \textbf{Case 2}. $\wh{f}_{01},\wh{f}_{10}\leq 0$ and $\wh{f}_{01},\wh{f}_{10}$ are not both zero.
 Without loss of generality, assume that $\wh{f}_{01}<0$, the argument is analogous if $\wh{f}_{10}<0$ instead.
 Thus, we have
 \begin{align}
  a+c &< b+d, \label{eq:ac_lt_bd} \\
  a+b &\leq c+d, \label{eq:ab_leq_cd}
 \end{align}
 as well as the lsm condition $ad>bc$.
 We will show the result by contradiction, i.e.\ assume for a contradiction that $\wh{f}_{11}<0$ or, equivalently,
 \begin{equation}
  a+d<b+c. \label{eq:ad_lt_bc}
 \end{equation}
 Adding the inequalities \eqref{eq:ac_lt_bd} and \eqref{eq:ab_leq_cd}, we have $a<d$.
 Adding \eqref{eq:ac_lt_bd} and \eqref{eq:ad_lt_bc} gives $a<b$, and adding up \eqref{eq:ab_leq_cd} and \eqref{eq:ad_lt_bc} we obtain $a<c$, with all of these inequalities being strict.
 Furthermore, the lsm condition $ad>bc$ now implies $b,c<d$.
 
 Next we show that $a^{2^k}+d^{2^k}<b^{2^k}+c^{2^k}$ for any $k$.
 Since $d>a,b,c$, this is a contradiction if $k$ is large enough. 
 We prove the property by induction on $k$.
 The base case, $k=0$, is $\wh{f}_{11}<0$.
 For the induction step, suppose $a^{2^k}+d^{2^k}<b^{2^k}+c^{2^k}$ for some $k\in\zN$.
 By squaring both sides of this inequality we get
 \begin{equation}\label{equ:1}
  a^{2^{k+1}}+2a^{2^k}d^{2^k}+d^{2^{k+1}} < b^{2^{k+1}}+2b^{2^k}c^{2^k}+c^{2^{k+1}}.
 \end{equation}
 Now, $ad>bc$ implies $a^{2^k}d^{2^k}>b^{2^k}c^{2^k}$. 
 Subtracting this from (\ref{equ:1}) we find
 \[
  a^{2^{k+1}}+d^{2^{k+1}} < b^{2^{k+1}}+c^{2^{k+1}},
 \]
 so the property does indeed hold for all $k$.
 Yet, since $d$ is strictly the largest of the four values, this inequality cannot be true for large $k$, a contradiction.
 Thus, the assumption $\wh{f}_{11}<0$ must have been false and in this case, too, we have $\wh{f}_{11}\geq 0$.
 
 \textbf{Case 3}. $\wh{f}_{01},\wh{f}_{10}\geq 0$ and $\wh{f}_{01},\wh{f}_{10}$ are not both zero.
 Let $f'(x,y)=\bar{f}(x,y)$, where we are renaming the function to avoid clumsy notation.
 Note that, if $f$ is nontrivial and lsm, then so is $f'$.
 By Observation~\ref{obs:Fourier_bit_flip}, we have $\wh{f'}_{00}=\wh{f}_{00}$, $\wh{f'}_{01}=-\wh{f}_{01}$, $\wh{f'}_{10}=-\wh{f}_{10}$, and $\wh{f'}_{11}=\wh{f}_{11}$.
 Now, $f'$ is a nontrivial binary lsm function with $\wh{f'}_{01},\wh{f'}_{10}\leq 0$ and $\wh{f'}_{01},\wh{f'}_{10}$ not both zero.
 Thus, by Case~2 of this proof, $\wh{f'}_{11}$ is nonnegative.
 This immediately implies that $\wh{f}_{11}\geq 0$.
 
 Thus, by exhaustive case analysis, we have shown that $\wh{f}_{11}\geq 0$ whenever $\wh{f}_{01}\wh{f}_{10}\geq 0$.
\end{proof}

\subsection{FPRAS for binary functions with nonnegative Fourier coefficients}

Recall from Section~\ref{s:Fourier} that $\cP$ is the set of functions in $\cB$ whose Fourier transform takes nonnegative values.
Recall 
from Section~\ref{s:counting_CSPs}  that, for a finite set of constraint functions $\cF$, the problem $\hol(\cF)$ is the restriction of $\NCSP(\cF)$ to instances where every variable appears exactly twice.

In  this section, we show that for all binary functions $f\in\cP\cap\cB_2$, the problem $\NCSP(f)$ has an FPRAS.
This result, Corollary~\ref{cor:Fourier_FPRAS} below,
arises as a corollary of a stronger statement (Theorem~\ref{thm:Fourier_FPRAS}).

An arity-$k$ function is said to be \emph{self-dual} if $f(x_1,\ldots,x_k) = f(\bar{x}_1,\ldots,\bar{x}_k)$.
Let $\SDP$ be the set of self-dual functions in~$\cP$.
The set $\SDP$,
introduced in~\cite{bulatov_functional_2017}, is a functional clone.
Let $\SDP_3 := \SDP\cap\cB_3$.
We show in  Theorem~\ref{thm:Fourier_FPRAS}
that, for any finite subset $\cF\sse\SDP_3$, the problem $\NCSP(\cF)$ has an FPRAS.
This is a somewhat surprising result because there are functions in $\SDP_3$
that are not log-supermodular. See \cite[Theorem 14]{bulatov_functional_2017}.

To prove  Theorem~\ref{thm:Fourier_FPRAS}, 
in Lemma~\ref{lem:ncsp_Holant}
we transform 
$\NCSP(\cF)$
to a suitable ``even subgraphs'' holant problem. In Lemma~\ref{lem:Holant_Fourier} we
find a bound on the weight of ``near-assignments'' of this holant problem.
In Lemma~\ref{lem:Holant_edge-weighted_graph} we
reduce the holant problem to a perfect matchings problem. In Lemma~\ref{lem:Holant_weight} we
bound the number of nearly perfect matchings in terms of near-assignments. Finally,
we apply a result of Jerrum and Sinclair (Lemma~\ref{lem:poly_approx}, from \cite{jerrum_approximating_1989})
to approximate the solution to the perfect matching problem.

Corollary~\ref{cor:Fourier_FPRAS} follows
from Theorem~\ref{thm:Fourier_FPRAS} via an AP-reduction $\NCSP(f)\leq_{AP} \NCSP(f')$, 
where $f'(x,y,z):=f(x\oplus z,y\oplus z)$ is shown to be in $\SDP_3$ for any $f\in\cP\cap\cB_2$.

Our proofs  will use the following characterisation
of $\SDP$.

\begin{lem}[{\cite[Lemma~38]{bulatov_functional_2017}}]\label{lem:SDP}
 Suppose $f\in\cB$, then $f\in\SDP$ if and only if the Fourier transform of $f$ is nonnegative on inputs of even Hamming weight and is zero on inputs of odd Hamming weight.
\end{lem}

As before, let $\oplus_3$ be the ternary indicator function for inputs of even parity:
\[
 \oplus_3(x,y,z) = \begin{cases} 1 &\text{if } x+y+z \text{ is even,} \\ 0 &\text{otherwise.} \end{cases}
\]
The following observation is well-known and arises frequently in holographic transformations.
See, e.g., \cite{Cai18}.
It also arises as a special case of \cite[Lemma 27]{bulatov_functional_2017}.

\begin{obs}\label{obs:FT_oplus3}
 The Fourier transform of $\oplus_3$ is $\wh{\oplus_3} = \frac{1}{2}\EQ_3$.
\end{obs}

\begin{lem}\label{lem:f_SDP}
 Suppose $f\in\cB_2$ and let $f'(x,y,z):=f(x\oplus z, y\oplus z)$.
 Then
 \[
  \wh{f'}(x,y,z)=\wh{f}(x,y)\cdot\oplus_3(x,y,z).
 \]
 Furthermore, if $f\in\cP$, then $f'\in\SDP_3$.
\end{lem}
\begin{proof}
 Let $g(x,y,z) := \wh{f}(x,y)\oplus_3(x,y,z)$, so in the first part of the lemma we are aiming to prove $\wh{f'}=g$.
 Note that by Observation~\ref{obs:Fourier_inverse} and the definition of the Fourier transform, we have
 \begin{equation}\label{eq:double-FT}
  f'(x,y,z) = 8 \wh{\wh{f'}}(x,y,z).
 \end{equation}
 By taking the Fourier transform, the desired equality $\wh{f'}=g$ is equivalent to $\wh{\wh{f'}}=\wh{g}$.
 Substituting into \eqref{eq:double-FT}, this is equivalent to $f'=8\wh{g}$.
 For the first part of the lemma, it thus suffices to show that the Fourier transform of $g$ is equal to $\frac{1}{8}f'$; this is more technically elegant than evaluating the Fourier transform of $f'$ directly.

 It will be useful to first define another function related to $f$.
 Let $h(x,y,z)=\wh{f}(x,y)$ be the function arising from $\wh{f}$ by introduction of a fictitious argument.
 By \cite[Lemma 25]{bulatov_functional_2017}, we have $\wh{h}(x,y,1) = 0$ and
 \[
  \wh{h}(x,y,0) = \wh{\wh{f\hphantom{'}}}(x,y) = \frac{1}{4}\sum_{p,q\in\{0,1\}} (-1)^{px+qy}\wh{f}(p,q) = \frac{1}{4} f(x,y),
 \]
 where the second equality is the definition of the Fourier transform of $\wh{f}$ and the third is by  Observation~\ref{obs:Fourier_inverse}.
 
 It is straightforward to see that $g(x,y,z) = h(x,y,z)\oplus_3(x,y,z)$.
 Thus, the Fourier transform of $g$ can be expressed as a convolution:\footnote{This is a well-known result in the theory of Fourier transforms. For pseudo-Boolean functions it follows from \cite{de_wolf_brief_2008} together with the straightforward-to-derive property that Fourier transforms are involutive up to scalar factor.}
 \[
  \wh{g}(x,y,z) = \sum_{p,q,r\in\{0,1\}} \wh{h}(p,q,r)\wh{\oplus_3}(p\oplus x, q\oplus y, r\oplus z).
 \]
 Substitute for $\wh{h}$, and for $\wh{\op_3}$ using Observation~\ref{obs:FT_oplus3}, then the Fourier transform becomes
 \[
  \wh{g}(x,y,z) = \sum_{p,q\in\{0,1\}} \frac{1}{4} f(p,q) \frac{1}{2} \EQ_3(p\oplus x, q\oplus y, z).
 \]
 Now, the ternary equality function enforces $p\op x = z$ and $q\op y=z$, which is equivalent to $p=x\op z$ and $q=y\op z$.
 Hence, the expression simplifies to
 \[
  \wh{g}(x,y,z) = \frac{1}{8} f(x\op z, y\op z) = \frac{1}{8} f'(x,y,z),
 \]
 the desired result.
 
 For the second part, note that $f\in\cP$ implies that $f$ has a nonnegative Fourier transform.
 Then, $\wh{f'}$ is the product of two nonnegative functions, so it is nonnegative, and therefore $f'\in\cP$.
 Yet $\wh{f'}$ is 0 on inputs of odd Hamming weight because of the factor of $\oplus_3$.
 Thus, by Lemma~\ref{lem:SDP}, $f'\in\SDP_3$.
\end{proof}

\begin{lem}\label{lem:f_SDP_reduction}
 Suppose $f\in\cB_2$ and $f'(x,y,z)=f(x\oplus z, y\oplus z)$, then $\NCSP(f)\leq_{AP}\NCSP(f')$.
\end{lem}
\begin{proof}
 Consider an instance $\Omega=(V,C)$ of $\NCSP(f)$.
 Suppose $V = \{x_1\zd x_n\}$ and label each constraint $((x_i,x_j),f)\in C$ by the tuple $(i,j)$ as a shorthand.
 This fully specifies the constraint as there is only one constraint function.
 Then
 \[
  Z(\Omega) = \sum_{\bx\in\{0,1\}^{n}} \prod_{(i,j)\in C} f(x_i,x_j).
 \]
 Because of the sum over all $\bx$, the following holds:
 \begin{align*}
  \sum_{\bx\in\{0,1\}^{n}} \prod_{(i,j)\in C} f(x_i,x_j)
  &= \frac{1}{2} \left( \sum_{\bx\in\{0,1\}^{n}} \prod_{(i,j)\in C} f(x_i,x_j) + \sum_{\bx\in\{0,1\}^{n}} \prod_{(i,j)\in C} f(x_i\oplus1,x_j\oplus1) \right) \\
  &= \frac{1}{2} \sum _{y\in\{0,1\}} \sum_{\bx\in\{0,1\}^{n}} \prod_{(i,j)\in C}   f(x_i\oplus y, x_j\oplus y).
 \end{align*}
 Now, the latter formula is $\frac{1}{2}$ times the partition function for an instance of $\NCSP(f')$ with set of variables $V'=\{x_1\zd x_n,y\}$ and set of constraints $C' = \{ ((x_i,x_j,y),f') \mid (i,j)\in C\}$.
 Thus we have $\NCSP(f)\leq_{AP}\NCSP(f')$.
\end{proof}

\begin{lem}\label{lem:ncsp_Holant}
 For any finite subset $\cF\sse\SDP_3$ we have $\NCSP(\cF)\leq_{AP}\hol(\wh{\cF},\op_3)$, where $\wh{\cF}=\{\wh{f}:f\in\cF\}$.
\end{lem}
\begin{proof}
 We will argue
 \[
  \NCSP(\cF) \leq_{AP} \hol(\cF,\EQ_3) \leq_{AP} \hol(\wh{\cF},\oplus_3).
 \]
 
 The first reduction is exactly a transformation to the holant framework; see Section~\ref{s:counting_CSPs} or \cite[Proposition 1]{cai_Holant_2012}.
 
 For the second reduction, let
 \[
  M= \begin{pmatrix}1&1\\1&-1\end{pmatrix}.
 \]
 Note that $M M^{T} = 2 I$. 
 The entry of the $k$-fold tensor product $M\otimes\ldots\otimes M$ corresponding to row $(\vc{x}{k})$ and column $(\vc{p}{k})$ is $  (-1)^{p_1x_1+\ldots+p_kx_k}$.
 Thus, for any $k$-ary function $h$,
 \[
  (M\circ h)(\vc{x}{k}) =   \sum_{\vc{p}{k}\in\{0,1\}} (-1)^{p_1x_1+\ldots+p_kx_k} h(\vc{p}{k}) = 2^{k}\wh{h}(\vc{x}{k}).
 \]
 Hence if $f$ has arity 3, then $M\circ f  = 2^3 \wh{f}$.
 In particular, $M\circ\EQ_3 = 2^3 \wh{\EQ_3} = 2\oplus_3$ by Observations~\ref{obs:Fourier_inverse} and \ref{obs:FT_oplus3}.
 
 Now by a corollary of Valiant's Holant Theorem, given here as Theorem~\ref{thm:valiant_Holant}, and the above connection between holographic transformations under $M$ and the Fourier transform, we have
 \[
  \hol(\cF,\EQ_3)=_{AP}\hol(M\circ\cF,M\circ\EQ_3) =_{AP} \hol\left(\cF', 2\op_3\right),
 \]
 where $\cF' = \{2^3\wh{f}: f\in\cF\}$.
 Recall that if a constraint function is multiplied by  a constant, this factor can be absorbed into an AP-reduction: hence 
 $\hol\left(\cF', 2 \op_3\right) =_{AP} \hol (\wh{\cF}, \op_3)$.
 This completes the chain of reductions.
\end{proof}

Recall from Section~\ref{s:counting_CSPs} that 
an assignment $\sigma$ of a holant instance $\Omega$ with variables $V$ and constraints $C$ is 
a function $\sigma\colon V \to \{0,1\}$
choosing a spin from $\{0,1\}$ for 
each variable in $V$. The assignment has a weight $w_\sigma$.

\begin{dfn}\label{dfn:near-assignment}
 A \emph{near-assignment} of a holant instance $\Omega=(V,C)$ is defined as follows: take two distinct variables $u,v\in V$.
 Replace the two occurrences of $u$ in constraints by two new variables $u',u''$ and add a new constraint $((u',u''),\NEQ)$.
 Similarly, replace the two occurrences of $v$ with new variables $v',v''$ and add a new constraint $((v',v''),\NEQ)$.
 This gives a new holant instance $\Omega_{u,v}$.
 Any  
 assignment~$\sigma$ of $\Omega_{u,v}$ 
 with $\sigma(u')\neq \sigma(u'')$ and $\sigma(v')\neq \sigma(v'')$  is called a near-assignment of $\Omega$.
\end{dfn}

Recall that $Z(\Omega)$ is the total weight of all assignments of
a holant instance $\Omega$ (see \eqref{eq:bit_flip_partition_function}).
Define
\begin{equation}\label{eq:near_assignments}
 Z'(\Omega) := \sum_{\{u,v\}\sse V} Z(\Omega_{u,v})
\end{equation}
to be the total weight of all near-assignments.

Note that assignments of $\Omega_{u,v}$ which are not near-assignments of $\Omega$ according to Definition~\ref{dfn:near-assignment} do not satisfy both disequality constraints.
The contribution of these assignments to the partition function $Z(\Omega_{u,v})$ is thus 0; hence $Z'(\Omega)$ is indeed the total weight of all near-assignments.

\begin{lem}\label{lem:Holant_Fourier}
 For any $n$-variable holant instance $\Omega$ which uses only functions in $\cP$, we have
 \[
  Z'(\Omega) \leq 2n^2 Z(\Omega).
 \]
\end{lem}
\begin{proof}
Let $\Omega=(V,C)$
and note that $|V|=n$.
Recall the definition of pps-formula  
(Definition~\ref{def:pps}).
Like any CSP instance, $\Omega$ can be viewed as
a pps-formula $\psi$ with no free variables. 
The atomic formulas of~$\psi$ correspond to the constraints  in~$C$.

Given any pair of  distinct variables $u,v\in V$,
let $\psi_{u,v}$ be the pps-formula 
obtained from~$\psi$  
by removing the (bound) variables $u$ and $v$
and introducing four new free variables, $w_1$, $w_2$, $w_3$ and $w_4$.
The two occurrences of $u$ in atomic formulas are replaced with $w_1$ and $w_2$.
The two occurrences of $v$ in atomic formulas are replaced with $w_3$ and $w_4$.

By construction, the function $f_{\psi_{u,v}}$
represented by the pps-formula $\psi_{u,v}$ 
(see the text below Definition~\ref{def:pps})
is in $\ang{\cP}$.
But $\cP$ is a functional clone
(see   
 Lemma~\ref{lem:Fou_omega}), so
$f_{\psi_{u,v}}(x_1,x_2,x_3,x_4) \in \cP$.

 The construction of $\psi_{u,v}$ guarantees
 that for any pair of distinct vertices $u,v$, we have
  \begin{equation}\label{eq:temptemp}
  Z(\Omega) = \sum_{x,y\in\{0,1\}}f_{\psi_{u,v}}(x,x,y,y).
  \end{equation}  
    Furthermore, 
 $$\sum_{x,y\in\{0,1\}} f_{\psi_{u,v}}(x,1-x,y,1-y)$$ is exactly the value $Z(\Omega_{u,v})$ arising in Definition~\ref{dfn:near-assignment}.
 Thus, by \eqref{eq:near_assignments}, the total weight of near-assignments of $\Omega$ is
 \begin{align*}
  Z'(\Omega) &= \sum_{\{u,v\}\sse V} \sum_{x,y\in\{0,1\}} f_{\psi_{u,v}}(x,1-x,y,1-y) \\
  &\leq \sum_{\{u,v\}\sse V} \sum_{x,y\in\{0,1\}} f_{\psi_{u,v}}(0,0,0,0) \\
  &\leq  \binom{n}{2} \cdot 4 \cdot Z(\Omega) \leq
  2n^2 Z(\Omega).
 \end{align*}
 Here, the  first inequality uses Lemma~\ref{lem:cP_000} and the second inequality uses~\eqref{eq:temptemp}.
\end{proof}

An \emph{edge-weighted multigraph} is a multigraph in which each edge $e$ is assigned a 
  non-negative rational weight 
 $w(e)$.
Given an edge-weighted multigraph $G$, a \emph{matching} of~$G$ is a subset $M\sse E(G)$ such that each vertex is incident to at most one edge in $M$.
The \emph{weight} of $M$, denoted $w(M)$, is the product of the weights of the edges in $M$.
A matching $M$ is \emph{perfect} if every vertex is incident to exactly one edge in $M$.
A matching $M$ is \emph{nearly perfect} if every vertex is incident to exactly one edge in $M$, except for two vertices which are not incident to any edges in $M$; the concept of near-perfect matchings will be important in Lemma~\ref{lem:poly_approx}.
The \emph{total weight of perfect matchings} of the multigraph $G$ is 
defined by $\ZPM(G) = \sum_M w(M)$,
where the sum is over all perfect matchings $M$ of $G$.
 The total weight of near-perfect matchings is  
 similarly defined as $\ZNPM(G) = \sum_M w(M)$, where
 the sum is over all near-perfect matchings $M$ of $G$.

Define $\wte\sse\cB_3$ to be the set of all ternary functions~$f$ with $f(0,0,0)=1$ and $f(0,0,1)=f(0,1,0)=f(1,0,0)=f(1,1,1)=0$.
\begin{obs}\label{obs:WtEven3}
 The set $\wte$ contains $\oplus_3$.
 Also, for every $f\in\SDP_3$, other than the all-zero function, $\wh{f}(0,0,0)>0$.
 Furthermore, by Lemma~\ref{lem:SDP}, $\wh{f}$ is zero on inputs of odd Hamming weight.
 Thus $\wte$ contains a scalar multiple of $\wh{f}$.
\end{obs}

An \emph{even assignment} of a holant instance $\Omega=(V,C)$ is an assignment $\sigma$ with the property that there is an even number of $1$ spins  in the scope of each constraint.
Formally, $\sigma:V\to\{0,1\}$ is an even assignment if, for all constraints $c\in C$, 
\[
 |\sigma(\bv_c)| := \sum_{w\in\bv_c} \sigma(w) \equiv 0 \pmod 2,
\]
where the sum is over all variables~$w$ in the tuple $\bv_c$ with their correct multiplicities.
If all constraint functions in the holant instance are taken from $\wte$, then only even assignments contribute to the partition function, as the contribution of each constraint is zero unless an even number of the variables in its scope are $1$.

Any holant instance using only constraint functions in $\wte$ can be transformed into the problem of counting perfect matchings on an edge-weighted graph with non-negative rational edge weights, as shown in the following definition and lemma.
This result is different from known results reducing the problem of computing holant values to counting perfect matchings on edge-weighted graphs, which generally use negative edge weights, even when the constraint functions satisfy parity conditions \cite[Lemma~2.25 \& Theorem~2.28]{curticapean_simple_2015}.

\begin{dfn}\label{dfn:Holant_edge-weighted_graph}
 Suppose that $\Omega=(V,C)$ is a holant instance using only constraint functions in $\wte$.
 Let $G''=(V'',E'')$ be the edge-weighted multigraph defined as follows:
 \begin{itemize}
  \item $G''$ contains three distinct vertices $c^{(1)}, c^{(2)}, c^{(3)}$ for each constraint $c\in C$.
  \item $G''$ contains the following weighted edges for each $c=(\bv_c,f_c)\in C$:
   \begin{itemize}
    \item an edge $\{c^{(1)}, c^{(2)}\}$ with weight $f_c(1,1,0)$,
    \item an edge $\{c^{(1)}, c^{(3)}\}$ with weight $f_c(1,0,1)$, and
    \item an edge $\{c^{(2)}, c^{(3)}\}$ with weight $f_c(0,1,1)$.
   \end{itemize}
   These edges are called \emph{edges within triangles}.
  \item $G''$ furthermore contains an edge $\{c^{(i)},d^{(j)}\}$ with weight $1$ for any pair of distinct constraints $c,d\in C$ and any $i,j\in\{1,2,3\}$ such that the same variable appears in the $i$-th position of $\bv_c$ and in the $j$-th position of $\bv_d$.
  Finally, $G''$ contains an edge $\{c^{(i)},c^{(j)}\}$ with weight $1$ for any $c\in C$ and any pair $i,j\in\{1,2,3\}$ with $i<j$ such that the same variable appears in both the $i$-th and $j$-th position of $\bv_c$.
  These edges are called \emph{edges between triangles}.
 \end{itemize}
\end{dfn}

Note that $G''$ may not be simple.  For example, if some variable appears in position~$1$ and position~$3$
of a constraint $c\in C$ then $G''$ has the edge $\{c^{(1)},c^{(3)}\}$ with weight $f_c(1,0,1)$
but also the edge $\{c^{(1)},c^{(3)}\}$ with weight~$1$. 

The ``triangle'' terminology is used because the graph $G''$ can be constructed by taking the graph representation $G'$ of $\Omega$ as defined in Observation~\ref{obs:Holant_graph} and then expanding each vertex of $G'$ into a triangle, as shown in Figure~\ref{fig:Holant_triangle}.
 
\begin{figure}
 \centering
  \begin{tikzpicture}
   \path (0,0) node[circle, inner sep=1pt, draw](F) {$f$} -- (0,1) node(x1) {$x_1$} -- (-.8,-.8) node(x2) {$x_2$} -- (.8,-.8) node(x3) {$x_3$};
   \draw (F) -- (x1);
   \draw (F) -- (x2);
   \draw (F) -- (x3);
   \draw[->] (1.5,0) -- (2.5,0);
   \draw (5.25,1) --  node[right] {$1$} ++(0,-0.5) -- node[right] {$f(1,0,1)$} ++(1,-1) -- node[right] {$1$} ++(.4,-.4);
   \filldraw (5.25,0.5) circle (3pt) -- node[left] {$f(1,1,0)$} ++(-1,-1) -- node[left] {$1$} ++(-.4,-.4);
   \filldraw (4.25,-.5) circle (3pt) -- node[below] {$f(0,1,1)$} ++(2,0) circle (3pt);
  \end{tikzpicture}
 \caption{A constraint $((x_1,x_2,x_3),f)$ in a holant problem using only constraint functions from $\wte$ corresponds to a triangle in the edge-weighted graph constructed in Definition~\ref{dfn:Holant_edge-weighted_graph}.}
 \label{fig:Holant_triangle}
\end{figure}
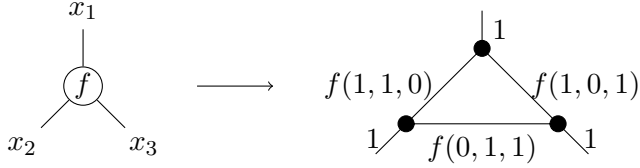

The following lemma relates the total weight of perfect matchings $\ZPM(G'')$
to the partition function $Z(\Omega)$ of the corresponding holant instance.

\begin{lem}\label{lem:Holant_edge-weighted_graph}
 Given any $n$-variable holant instance $\Omega$ using functions in $\wte$, let $G''$ be the corresponding edge-weighted multigraph according to Definition~\ref{dfn:Holant_edge-weighted_graph}.
 Then $\ZPM(G'') =  Z(\Omega)$.
\end{lem}
\begin{proof}
 Let $\Omega=(V,C)$ and $G''=(V'',E'')$.
 Note that all assignments contributing a non-zero weight to the partition function $Z(\Omega)$ are even since the constraint functions have support only on inputs of even Hamming weight.
 
 We define a weight-preserving bijection between the sets
 \begin{align*}
  \cE &= \{\sigma : \sigma \text{ is an even assignment of } \Omega\} \quad\text{and} \\
  \cM &= \{M : M \text{ is a perfect matching of } G''\}.
 \end{align*}
 Consider $\sigma\in\cE$.
 For every variable $v\in V$, there is exactly one between-triangles edge $e_v=\{c^{(i)},d^{(j)}\}\in E''$ where $c,d\in C$ such that $v$ appears in the $i$-th position  of $\bv_c$ and in the $j$-th position  of $\bv_d$.
 Let $S_\sigma = \{e_v:\sigma(v) = 0\}$.
 This set is a matching because each vertex in $V''$ is incident on exactly one between-triangles edge, and $S_\sigma$ is a subset of the edges between triangles.
 Let $T_\sigma$ be the following subset of within-triangle edges:
 \[
  T_\sigma = \left\{ \{c^{(i)},c^{(j)}\} : c\in C \text{ with } \sigma(\bv_c[i])=\sigma(\bv_c[j])=1 \text{ and } 1\leq i<j\leq 3 \right\}.
 \]
 
 A portion of the set $M_\sigma=S_\sigma\cup T_\sigma$
 is illustrated in Figure~\ref{fig:assignment_matching}. We will show that
 $M_\sigma$ is a perfect matching.
 To see this, consider some vertex $c^{(i)}\in V''$, which corresponds to a variable $\bv_c[i]$.
 If $\sigma(\bv_c[i])=0$ then $c^{(i)}$ is matched by an edge in $S_\sigma$, this edge is unique in $S_\sigma$ as there is only one edge between triangles incident on any given vertex.
 Furthermore, $c^{(i)}$ cannot appear in any edges in $T_\sigma$.
 If $\sigma(\bv_c[i])=1$ then $c^{(i)}$ cannot appear in any edges in $S_\sigma$.
 Yet since $\sigma$ is even, there must be a unique $j \in\{1,2,3\}\setminus\{i\}$ such that $\sigma(\bv_c[j])=1$, and thus a unique edge in $T_\sigma$ which is incident on $c^{(i)}$.
 Thus $M_\sigma$ is indeed a perfect matching.

 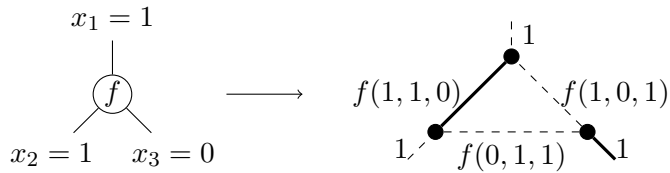
\begin{figure}
  \centering
  \begin{tikzpicture}
   \path (0,0) node[circle, inner sep=1pt, draw](F) {$f$} -- (0,1) node(x1) {$x_1=1$} -- (-.8,-.8) node(x2) {$x_2=1$} -- (.8,-.8) node(x3) {$x_3=0$};
   \draw (F) -- (x1);
   \draw (F) -- (x2);
   \draw (F) -- (x3);
   \draw[->] (1.5,0) -- (2.5,0);
   \path (5.25,1.1) node(e1) {} -- (5.25,.5) node[circle, inner sep=2pt, draw, fill](c1) {} -- (6.25,-.5) node[circle, inner sep=2pt, draw, fill](c3) {} -- (6.75,-1) node(e3) {} -- (3.75,-1) node(e2) {} -- (4.25,-.5) node[circle, inner sep=2pt, draw, fill](c2) {};
   \draw[dashed] (e1) -- node[right] {$1$} (c1);
   \draw[dashed] (e2) -- node[left] {$1$} (c2);
   \draw[very thick] (e3) -- node[right] {$1$} (c3);
   \draw[very thick] (c1) -- node[left] {$f(1,1,0)$} (c2);
   \draw[dashed] (c2) -- node[below] {$f(0,1,1)$} (c3);
   \draw[dashed] (c3) -- node[right] {$f(1,0,1)$} (c1);
  \end{tikzpicture}
  \caption{Suppose that the assignment~$\sigma$ 
  maps  $(x_1,x_2,x_3)\mapsto(1,1,0)$.
  The thickened edge labelled ``$1$'' is in $S_\sigma$ and the other thickened edge is in $T_\sigma$.
  The dashed edges are not in $M_\sigma$.}
  \label{fig:assignment_matching}
 \end{figure} 
  
 Conversely, given $M\in\cM$, define an assignment of $\Omega$ based on whether $e_v$, the unique between-triangles edge corresponding to variable $v$, is in $M$:
 \[
  \sigma_M(v) = \begin{cases} 0 &\text{if } e_v\in M, \\ 1 &\text{otherwise.} \end{cases}
 \]
 This definition fully specifies $\sigma_M$ and each assignment specified in this way is even.
 The latter property arises because for $M$ to be a perfect matching each triangle must satisfy one of the following two properties, cf.\ Figure~\ref{fig:triangle_matchings}:
 \begin{itemize}
  \item either, none of the within-triangle edges and all three adjacent between-triangle edges are in $M$, or
  \item one of the within-triangle edges and one of the adjacent between-triangle edges are in $M$.
 \end{itemize}
 A matching can contain at most one edge of a triangle, so this covers all cases.
 Thus, $\sigma_M\in\cE$.
 
 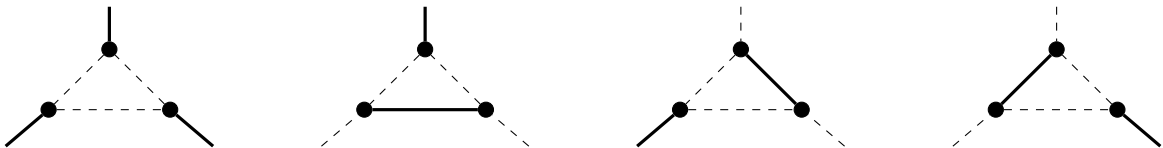
\begin{figure}
  \centering
  \begin{tikzpicture}
   \path (5.25,1.1) node(e1) {} -- (5.25,.4) node[circle, inner sep=2pt, draw, fill](c1) {} -- (6.05,-.4) node[circle, inner sep=2pt, draw, fill](c3) {} -- (6.75,-1) node(e3) {} -- (3.75,-1) node(e2) {} -- (4.45,-.4) node[circle, inner sep=2pt, draw, fill](c2) {};
   \draw[very thick] (e1) -- (c1);
   \draw[very thick] (e2) -- (c2);
   \draw[very thick] (e3) -- (c3);
   \draw[dashed] (c1) -- (c2);
   \draw[dashed] (c2) -- (c3);
   \draw[dashed] (c3) -- (c1);
  \end{tikzpicture}
  \qquad
  \begin{tikzpicture}
   \path (5.25,1.1) node(e1) {} -- (5.25,.4) node[circle, inner sep=2pt, draw, fill](c1) {} -- (6.05,-.4) node[circle, inner sep=2pt, draw, fill](c3) {} -- (6.75,-1) node(e3) {} -- (3.75,-1) node(e2) {} -- (4.45,-.4) node[circle, inner sep=2pt, draw, fill](c2) {};
   \draw[very thick] (e1) -- (c1);
   \draw[dashed] (e2) -- (c2);
   \draw[dashed] (e3) -- (c3);
   \draw[dashed] (c1) -- (c2);
   \draw[very thick] (c2) -- (c3);
   \draw[dashed] (c3) -- (c1);
  \end{tikzpicture}
  \qquad
  \begin{tikzpicture}
   \path (5.25,1.1) node(e1) {} -- (5.25,.4) node[circle, inner sep=2pt, draw, fill](c1) {} -- (6.05,-.4) node[circle, inner sep=2pt, draw, fill](c3) {} -- (6.75,-1) node(e3) {} -- (3.75,-1) node(e2) {} -- (4.45,-.4) node[circle, inner sep=2pt, draw, fill](c2) {};
   \draw[dashed] (e1) -- (c1);
   \draw[very thick] (e2) -- (c2);
   \draw[dashed] (e3) -- (c3);
   \draw[dashed] (c1) -- (c2);
   \draw[dashed] (c2) -- (c3);
   \draw[very thick] (c3) -- (c1);
  \end{tikzpicture}
  \qquad
  \begin{tikzpicture}
   \path (5.25,1.1) node(e1) {} -- (5.25,.4) node[circle, inner sep=2pt, draw, fill](c1) {} -- (6.05,-.4) node[circle, inner sep=2pt, draw, fill](c3) {} -- (6.75,-1) node(e3) {} -- (3.75,-1) node(e2) {} -- (4.45,-.4) node[circle, inner sep=2pt, draw, fill](c2) {};
   \draw[dashed] (e1) -- (c1);
   \draw[dashed] (e2) -- (c2);
   \draw[very thick] (e3) -- (c3);
   \draw[very thick] (c1) -- (c2);
   \draw[dashed] (c2) -- (c3);
   \draw[dashed] (c3) -- (c1);
  \end{tikzpicture}
  \caption{The different matchings for a triangle, where thick lines denote edges in the matching and dashed lines denote edges not in the matching.}
  \label{fig:triangle_matchings}
 \end{figure}
 
 It is straightforward to see that the maps $\sigma\mapsto M_\sigma$ and $M\mapsto\sigma_M$ are inverse to each other.
 Hence they define a bijection between $\cE$ and $\cM$.
 It remains to show this bijection is weight-preserving.
 
 Consider $\sigma\in\cE$, then $w(\sigma)=\prod_{c\in C} f_c(\sigma(\bv_c))$.
 Recall that all edges between triangles have weight 1, so
 \[
  w(M_\sigma) = \prod_{e\in M_\sigma} w(e) = \prod_{e\in T_\sigma} w(e).
 \]
 Observe from Definition~\ref{dfn:Holant_edge-weighted_graph} that for any $\{c^{(i)},c^{(j)}\}\in T_\sigma$, $w(\{c^{(i)},c^{(j)}\})=f_c(\sigma(\bv_c)$.
 Let $C'=\{c\in C: \abs{\sigma(\bv_c)}=2\}$, where $\abs{\cdot}$ denotes the Hamming weight, then $w(M) = \prod_{c\in C'} f_c(\sigma(\bv_c))$.
 But $\sigma$ is an even assignment and $f_c(0,0,0)=1$ for all $c$.
 Thus $w(\sigma)=w(M_\sigma)$.
\end{proof}

\begin{lem}\label{lem:Holant_weight}
 Given any $n$-variable holant instance $\Omega$ using functions in $\wte$, 
 let $G''$ be the corresponding edge-weighted multigraph according to Definition~\ref{dfn:Holant_edge-weighted_graph}.
 Then 
 $\ZNPM(G'') \leq  nZ(\Omega)+Z'(\Omega)$.
\end{lem}
\begin{proof}
Let $\Omega=(V,C)$ where $n=|V|$.
We now define three sets.
Let 
\begin{align*}
\cM &= \{M \mid \text{$M$ is a near-perfect matching of $G''$ with $w(M)>0$}\}.\\
A &= \{(v,\sigma) \mid \text{$v\in V$ and $\sigma$  is an even assignment of $\Omega$ with $\sigma(v)=0$}\}, \text{and}\\
B &= \{(u,v,\sigma) \mid \text{$u\in V$ and $v\in V$ are distinct and
$\sigma$ is an assignment of $\Omega_{u,v}$ that is} \\
&\quad \quad \text{even at all constraints
except }   ((u',u''),\NEQ) \text{ and } ((v',v''),\NEQ) \}.
\end{align*}

Since the only assignments of $\Omega$ that contribute to $Z(\Omega)$ are
even assignments, we have $\sum_{(v,\sigma)\in A } w_\sigma \leq n Z(\Omega)$.
Similarly, since the only assignments 
of $\Omega_{u,v}$ that contribute to $Z(\Omega_{u,v})$ are even at all ternary constraints,
$\sum_{(u,v,\sigma)\in B} w_\sigma = Z'(\Omega)$.
We will define a weight-preserving injection~$\tau$ from $\cM$ into $A \cup B$.

Consider $M\in \cM$.
Every vertex of $G''$ is incident to exactly one edge in $M$
except for two vertices, which are not incident to any edges in $M$.
Say that these two vertices are $c_1^{(i_1)}$ and $c_2^{(i_2)}$ for
some $c_1,c_2\in C$ and $i_1,i_2 \in \{1,2,3\}$.
Note that $c_1$ might coincide with $c_2$
and $i_1$ might coincide with $i_2$, but
$c_1^{(i_1)} \neq c_2^{(i_2)}$.

Now, for every variable
$v\in V$ there is exactly one
between-triangle edge, say $\{c_3^{(i_3)},c_4^{(i_4)}\}$,
such that $v$ occurs in position
position~$i_3$ of~$c_3$
and position~$i_4$ of   $c_4$.
The idea will be to define $\sigma_M(v)\in \{0,1\}$
based on whether this edge is in~$M$.

First, if $\{c_3^{(i_3)},c_4^{(i_4)}\}$ 
is disjoint from $\{c_1^{(i_1)},c_2^{(i_2)}\}$,
we make the following definition, which is similar to the construction in the proof of 
Lemma~\ref{lem:Holant_edge-weighted_graph}:
\begin{equation}\label{eq:toinvert}
\sigma_M(v) = 
\begin{cases}
0, & \text{if $\{c_3^{(i_3)},c_4^{(i_4)}\}\in M$,}\\
1, & \text{otherwise.} \end{cases}\end{equation}

Now recall that $c_1^{(i_1)}$ and $c_2^{(i_2)}$
are not matched in $M$.
There are two cases to consider, 
depending on whether 
$\{c_1^{(i_1)},c_2^{(i_2)}\}$ is an edge of $G''$ or not. 
In each case, we define  $\tau(M)$ and argue that it is weight-preserving.
Later, we will argue that $\tau$ is an injection.

\textbf{Case 1}. Suppose that $\{c_1^{(i_1)},c_2^{(i_2)}\}$ is an edge of $G''$.
 In this case, there is exactly one variable $v$ for
 which $\sigma_M(v)$ is not  defined by~\eqref{eq:toinvert}
 and this variable $v$   is in  position~$i_1$ of~$c_1$ and position~$i_2$ of~$c_2$.
So define $\sigma_M(v) = 0$
and define $\tau(M) = (v,\sigma_M)$.
The fact that $w(M) = w_{\sigma_M}$
is similar to the argument that we already made in 
the proof of Lemma~\ref{lem:Holant_edge-weighted_graph}.
Let $M'$ be the perfect matching of $G''$ consisting of
$M$ and the edge $\{ c_1^{(i_1)}, c_2^{(i_2)}\}$.
Since $\{c_1^{(i_1)},c_2^{(i_2)}\}$ is a between-triangle edge (with weight~$1$), we
have $w(M')=w(M)$.
Since $\sigma_M(v)=0$, 
the assignment $\sigma_M$ coincides with the assignment $\sigma_{M'}$
constructed in the proof of Lemma~\ref{lem:Holant_edge-weighted_graph} so
$w_{\sigma_M} = w_{\sigma_{M'}}$.
But we have already argued, in the proof of Lemma~\ref{lem:Holant_edge-weighted_graph}
that $\sigma_{M'} = w(M')$. 
So we have proved that $w(M)=w_{\sigma_M}$. 
Since $w(M)>0$, this ensures that 
 $\sigma_M$ is an even assignment, so 
 $\tau(M)\in A$. 
 
\textbf{Case 2}. Suppose that $\{c_1^{(i_1)},c_2^{(i_2)}\}$ is not an edge of $G''$.
Then there are exactly two variables $u_1$ and $u_2$ for which 
$\sigma_M(u_1)$ and $\sigma_M(u_2)$  are not defined by~\eqref{eq:toinvert}.
Variable $u_1$ appears in position $i_1$ of $c_1$  (and somewhere else). 
Variable $u_2$ appears in position $i_2$ of $c_2$ (and somewhere else). 
Consider the instance $\Omega_{u_1,u_2}$ constructed as in Definition~\ref{dfn:near-assignment}. 
Suppose that the variables replacing $u_1$ in $\Omega_{u_1,u_2}$ are $u_1'$ and $u_1''$
with $u_1'$ appearing in $c_1^{(i_1)}$ 
and
that the  variables replacing $u_2$ in $\Omega_{u_1,u_2}$ are $u_2'$ and $u_2''$ 
with $u_2'$ appearing in $c_2^{(i_2)}$.
Then set $\sigma_{M}(u_1')=\sigma_{M}(u_2')=0$ and $\sigma_{M}(u_1'')=\sigma_{M}(u_2'')=1$.
This gives an assignment for $\Omega_{u_1,u_2}$ which satisfies the disequality constraints and which has an even number of spin-1 variables in each ternary constraint.
So define $\tau(M) =   (u_1,u_2,\sigma_M)$ and note that $\tau(M)\in B$
and   
(using the same arguments as in the proof of Lemma~\ref{lem:Holant_edge-weighted_graph}) 
that
$w(M) = \sigma_M$.
   
To conclude the proof, we must argue that $\tau$ is an injection.
This is straightforward. From $(v,\sigma)\in A$
there is a unique edge 
$\{c_1^{(i_1)},c_2^{(i_2)}\}$ of $G''$ corresponding to~$v$.
Leave this out of $M$ and recover the rest of 
the intersection of $M$ 
and the between-triangle edges of~$G''$
using Equation~\eqref{eq:toinvert}.
There is a unique extension to  the edges within triangles
which gives a near-perfect matching of $G''$ 
where $c_1^{(i_1)}$ and $c_2^{(i_2)}$ are unmatched.
Similarly,
given $(u_1,u_2,\sigma) \in B$ 
it is easy to identify the between-triangle edges of $G''$ corresponding to $u_1$ and $u_2$.
Leave $u_1$ and $u_2$ unmatched in~$M$ and recover the rest of
the between-triangle edges of~$M$ from~\eqref{eq:toinvert}.
Again, there is a unique extension to edges within triangles.
 \end{proof}

\begin{obs}\label{lem:weights_multigraph}
Let $G=(V,E)$ be an edge-weighted multigraph 
in which each edge $e\in E$ has a non-negative rational weight $w(e)$.
Let $d$ be the least common   denominator of the positive edge weights of~$G$.
Let $G'=(V,E')$ be the \emph{unweighted} multigraph 
defined as follows. 
For each $\{u,v\} \subseteq V$, let $\{e_1,\ldots,e_{k_{u,v}}\}$ be the set of edges from $u$ to $v$ in $E$.
The number of edges $\{u,v\}$ in 
$E'$ is defined to be $d \sum_{j=1}^{k_{u,v}} w(e_j)$.
Similarly, for each $v\in V$, let $\{e_1,\ldots,e_{k_v}\}$ be the set of self-loops on $v$ in $E$.
The number of self-loops on $v$ in $E'$ is defined to be 
$d \sum_{j=1}^{k_v} w(e_j)$.
Then   
$\ZPM(G') = d^{\abs{V}/2} \ZPM(G)$
and
$\ZNPM(G') = d^{\abs{V}/2-1} \ZNPM(G)$ 
\end{obs}
\begin{proof}
Let $n = \abs{V}$.
Let $H$ be an edge-weighted multigraph with the same
vertices and edges as $G$ 
except that the weight of each edge $e$ is $d w(e)$.
Since each perfect matching of $G$ has $n/2$ edges
and each near-perfect matching has $n/2-1$ edges,
 $\ZPM(H) = d^{ n/2} \ZPM(G)$
and
$\ZNPM(H) = d^{ n/2-1} \ZNPM(G)$. 
Note that all of the edge-weights of $H$ are non-negative integers.
Then it is easy to see that 
$\ZPM(G') = \ZPM(H)$ and $\ZNPM(G') = \ZNPM(H)$.
\end{proof}

Observation~\ref{lem:weights_multigraph} allows us to use the following result, where $M_n(G)$ denotes the number of perfect matchings of a $2n$-vertex multigraph $G$, and $M_{n-1}(G)$ is the number of near-perfect matchings of $G$.
While the result was originally stated for graphs, its proof also applies to   multigraphs.  

\begin{lem}[{\cite[Corollary 3.7]{jerrum_approximating_1989}}]\label{lem:poly_approx}
 Let $q$ be any fixed polynomial.
 There exists an FPRAS for $\abs{M_n(G)}$ 
 when the input is restricted to be a
  $2n$-vertex multigraph   $G$  satisfying $\abs{M_{n-1}(G)} \leq q(n) \abs{M_n(G)}$.
\end{lem}

\begin{thm}\label{thm:Fourier_FPRAS}
 Suppose that $\cF$ is a finite subset of $\SDP_3$.
 Then $\NCSP(\cF)$ has an FPRAS.
\end{thm}
\begin{proof}
 If $\cF$ contains the all-zero function $f_0$, then any instance $\Omega$ of $\NCSP(\cF)$ which contains a constraint using $f_0$ satisfies $Z(\Omega)=0$.
 The property of whether $\Omega$ contains a constraint using $f_0$ can be checked in polynomial time.
 Thus, $\NCSP(\cF)\leq_{AP}\NCSP(\cF\setminus\{f_0\})$.
 In other words, it suffices to consider sets $\cF$ that do not contain the all-zero function.

 By Lemma~\ref{lem:ncsp_Holant}, we have $\NCSP(\cF)\leq_{AP}\hol(\wh{\cF},\op_3)$, where $\wh{\cF}=\{\wh{f}:f\in\cF\}$.
 For any $f\in\cF$, since $f$ is not the all-zero function and all values are nonnegative,
 \[
  c_f := \wh{f}(0,0,0) = \frac{1}{8}\sum_{x,y,z\in\{0,1\}} f(x,y,z)>0.
 \]
 Let $f''(x,y,z) = c_f^{-1} \wh{f}(x,y,z)$ be a normalised version of the Fourier transform of $f$, and let $\cF''=\{f'': f\in\cF\}$.
 Then $\hol(\wh{\cF},\op_3)=_{AP}\hol(\cF'',\op_3)$.
 Now, $\cF\sse\SDP_3$, so by Observation~\ref{obs:WtEven3}, $\cF''\cup\{\oplus_3\}\sse\wte$. Hence, given any instance $\Omega$ of the problem $\hol(\cF'',\op_3)$ 
 we can 
 use  Definition~\ref{dfn:Holant_edge-weighted_graph} to
 construct an edge-weighted multigraph $G''$ so
 that, by Lemma~\ref{lem:Holant_edge-weighted_graph},
 $Z(\Omega) = \ZPM(G'')$. 
 Let $n$ be the number of variables of $\Omega$.
 To avoid trivialities, we assume $n\geq 1$.
  
 By Lemma~\ref{lem:Holant_weight},
 $\ZNPM(G'') \leq  n Z(\Omega)+Z'(\Omega)$
 and by Lemma~\ref{lem:Holant_Fourier},
 $Z'(\Omega)  \leq 2 n^2 Z(\Omega)$
 so $\ZNPM(G'') \leq   3n^2  \ZPM(G'')$.
 
 Using Observation~\ref{lem:weights_multigraph},
 it is easy to define an unweighted multigraph $G'$ 
 such that 
 $\ZPM(G') = d^{\abs{V(G'')}/2} \ZPM(G'')$ and $\ZNPM(G') = d^{\abs{V(G'')}/2-1} \ZNPM(G'')$, where
 $d$ is the least common denominator of the positive edge weights of $G''$.
 
 The FPRAS for computing $\ZPM(\cdot)$ 
 from Lemma~\ref{lem:poly_approx} 
 can be used to approximate   $\ZPM(G')$  which gives an approximation to 
 $\ZPM(G'')=Z(\Omega)$.
 
 As $\NCSP(\cF)\leq_{AP}\hol(\cF'',\op_3)$, this implies the existence of an FPRAS for $\NCSP(\cF)$.
\end{proof}

\begin{cor}\label{cor:Fourier_FPRAS}
 For all functions $f\in\cP\cap\cB_2$, the problem $\NCSP(f)$ has an FPRAS.
\end{cor}
\begin{proof}
 Let $f'(x, y, z) := f(x \oplus z, y \oplus z)$, then $\NCSP(f)\leq_{AP}\NCSP(f')$ by Lemma~\ref{lem:f_SDP_reduction}.
 Furthermore, $f'\in\SDP_3$ by Lemma~\ref{lem:f_SDP}, so $\NCSP(f')$ has an FPRAS by Theorem~\ref{thm:Fourier_FPRAS}.
 Therefore, there exists an FPRAS for $\NCSP(f)$.
\end{proof}

\subsection{Non-monotone nontrivial non-lsm functions}

The remaining case in the statement of Theorem~\ref{thm:two-spin} is that of a nontrivial non-lsm non-monotone function.

\begin{lem}\label{lem:nnn_NP}
 If $f\in\cB_2$ is nontrivial and non-lsm, and both $f$ and $\bar{f}$ are non-monotone, then $\NCSP(f)$ does not have an FPRAS unless $\NP=\RP$.
\end{lem}

\begin{proof}
 Suppose $f$ is nontrivial and non-lsm, then $ad<bc$ by Observation~\ref{obs:binary_lsm} and $a,d$ are not both zero.
 
 If $\wh{f}_{01}\wh{f}_{10}<0$, then both $\ang{f}\cap\cB^{<}_1$ and $\ang{f}\cap\cB^{>}_1$ are non-empty by Lemma~\ref{lem:make_up_down}.
 Hardness then follows by Theorem~\ref{thm:up-down} and Lemma~\ref{lem:clone_csp}, noting that the reduction used in the latter result is actually a simulation (cf.\ the proof of Lemma~\ref{lem:parity_hard}).
 
 Hence, suppose instead that $\wh{f}_{01}\wh{f}_{10}\geq 0$.
 Without loss of generality, assume both Fourier coefficients are nonpositive, i.e.
 \begin{align}
  a+c &\leq b+d, \label{eq:F01}\\
  a+b &\leq c+d. \label{eq:F10}
 \end{align}
 (If instead both Fourier coefficients are nonnegative, replace $f$ with $f':=\bar{f}$, then $\NCSP(f')=_{AP}\NCSP(f)$ by Observation~\ref{obs:bit_flip_CSP} and $\wh{f'}_{01}, \wh{f'}_{10}$ are both nonpositive by Observation~\ref{obs:Fourier_bit_flip}.
 Furthermore, the assumptions of the lemma remain unchanged since both $f$ and $\bar{f}$ are non-monotone and the properties of being nontrivial and non-lsm are invariant under bit-flips.)
 
 Adding \eqref{eq:F01} and \eqref{eq:F10} gives
 \begin{equation}\label{eq:a_leq_d}
  a\leq d.
 \end{equation}
 The function $f$ is monotone if $a\leq b\leq d$ and $a\leq c\leq d$.
 Without loss of generality, it suffices to consider two non-monotone cases: $b<a$ or $d<b$; if only $c$ fails to satisfy monotonicity, replace $f(x,y)$ by $f(y,x)$ and proceed as before.
 By Observation~\ref{obs:Fourier_swap}, swapping the variables swaps the Fourier coefficients $\wh{f}_{01}$ and $\wh{f}_{10}$, so the assumption that both are nonpositive remains valid.

 If $b<a$, then \eqref{eq:F01} implies that $c<d$.
 But multiplying these two inequalities yields $bc < ad$, contradicting the assumption that $f$ is non-lsm.

 Thus the only case to consider is
 \begin{equation}\label{eq:d_lt_b}
  d<b,
 \end{equation}
 which together with \eqref{eq:F10} implies
 \begin{equation}\label{eq:a_lt_c}
  a<c.
 \end{equation}
 We distinguish two subcases.
 \begin{enumerate}
  \item Suppose that equality holds in both \eqref{eq:F01} and \eqref{eq:F10}.
   Then, by adding or subtracting the two equations, we find $a=d$ and $b=c$.
   Since $f$ is nontrivial, $b\neq 0$, so let $f':=\frac{1}{b}\cdot f$.
   As constant factors can be absorbed into AP-reductions, we have $\NCSP(f) =_{AP} \NCSP(f')$.
   Now, $f'$ is an antiferromagnetic Ising function and, by \cite[Theorem 3]{goldberg_computational_2003}, $\NCSP(f')$ does not have an FPRAS unless $\NP=\RP$.
   
  \item We may now assume that at least one of \eqref{eq:F01} and \eqref{eq:F10} is a strict inequality.
   If \eqref{eq:F01} is strict, let $\upf(x) := \sum_y f(y,x)$, which satisfies $\upf(0)=a+c<b+d=\upf(1)$.
   Otherwise, \eqref{eq:F10} is strict, so $\upf(x) := \sum_y f(x,y)$ satisfies $\upf(0)=a+b<c+d=\upf(1)$.
   In either case, $\upf$ is a strictly increasing unary function in $\ang{f}$.
   It is also permissive as \eqref{eq:a_leq_d}, \eqref{eq:d_lt_b} and \eqref{eq:a_lt_c}, together with nonnegativity of all values, imply that $(a+b)$ and $(a+c)$ are both strictly positive.
   By Lemma~\ref{lem:clone_csp}, $\NCSP(f,\upf) \leq_{AP} \NCSP(f)$.
   
   Now, according to Lemma~\ref{lem:normalising}, there exists $\upf'\in\cB^{<,\mathrm{n}}_1$ such that $\NCSP(f,\upf') \leq_{AP} \NCSP(f,\upf) \leq_{AP} \NCSP(f)$.
   It thus suffices to show that $\NCSP(f,\upf')$ does not have an FPRAS unless $\NP=\RP$.
   
   Following from Observation~\ref{obs:pinning}, $\delta_1\in\ang{f,\upf'}_{\om,p}$.
   We can therefore pin inputs to the value 1.
   In particular, $\downf(x):=f(x,1)$ is a unary function in $\ang{f,\upf'}_{\om,p}$; it is strictly decreasing as $f(0,1)=b>d=f(1,1)$ by \eqref{eq:d_lt_b}.
   It remains to check that $\downf$ is permissive.
   To see this, note that $f$ is nontrivial and $0\leq a\leq d$.
   Together, these two properties imply that $d>0$, because $0=a=d$ would make $f$ trivial.
   So $\downf$ is indeed permissive.
   Using Lemma~\ref{lem:clone_csp} again, we find that $\NCSP(f,\upf',\downf) \leq_{AP} \NCSP(f)$.
   But, by Theorem~\ref{thm:up-down}, $\NCSP(f,\upf',\downf)$ does not have an FPRAS unless $\NP=\RP$; the same condition thus applies to $\NCSP(f)$.
 \end{enumerate}  
 By exhaustive analysis of all cases, we have therefore shown that, whenever $f$ is a nontrivial non-lsm function and both $f$ and $\bar{f}$ are non-monotone, $\NCSP(f)$ does not have an FPRAS unless $\NP=\RP$.
\end{proof}

\subsection{Proof of Theorem \ref{thm:two-spin}}
\label{s:two-spin-proof}

We now have all the pieces required to prove the main theorem of this section.

\begin{repthm}{thm:two-spin}
 \statetwospin{}{}{}{}{}
\end{repthm}

The only cases not covered by Theorem~\ref{thm:two-spin} are when either function $f(x,y)$ or function $f(1-x,1-y)$ is monotone. We address these cases to some extent in the next section.

\begin{proof}
 For Property~\ref{p:trivial}, note that any trivial binary function $f$ is in $\cN$, and $\NCSP(f)$ has an FPRAS by Theorem~\ref{thm:complexity_conservative}.
 
 If $f$ is a nontrivial binary lsm function whose Fourier coefficients $\wh{f}_{01}$ and $\wh{f}_{10}$ have opposite signs, then there exists $\upf\in\ang{f}\cap\cB^{<}_1$ and $\downf\in\ang{f}\cap\cB^{>}_1$ by Lemma~\ref{lem:make_up_down}.
 Thus, $\NCSP(f)$ is \BIS-hard by Lemma~\ref{lem:lsm-nontrivial}.
 It is also \BIS-easy by \cite[Theorem 47]{chen_complexity_2015}.
 This establishes Property~\ref{p:lsm_BIS-compl}.
 
 If $f$ is a nontrivial binary lsm function whose Fourier coefficients $\wh{f}_{01}$ and $\wh{f}_{10}$ are both nonnegative or both nonpositive, then by Observations~\ref{obs:Fourier_bit_flip} and \ref{obs:bit_flip_CSP}, it suffices to consider the case $\wh{f}_{01},\wh{f}_{10}\geq 0$: otherwise, replace $\NCSP(f)$ with $\NCSP(\bar{f})$, which is AP-interreducible with the former.
 Now, $\wh{f}_{00}\geq 0$ for all $f\in\cB_2$ by the definition of the Fourier transform, and $\wh{f}_{11}\geq 0$ for the given $f$ by Lemma~\ref{lem:Fourier_non-zero}, so all Fourier coefficients of $f$ are nonnegative.
 The problem $\NCSP(f)$ thus has an FPRAS by Corollary~\ref{cor:Fourier_FPRAS}, which proves Property~\ref{p:lsm_FPRAS}.
 
 If none of the above cases apply, $f$ is nontrivial and non-lsm.
 Now, if both $f$ and $\bar{f}$ are non-monotone, then, by Lemma~\ref{lem:nnn_NP}, $\NCSP(f)$ does not have an FPRAS unless $\NP=\RP$.
 This is Property~\ref{p:non-monotone}.
 
 We have thus established all the desired properties.
\end{proof}

With the current state of knowledge, it is possible to determine the complexity of $\NCSP(f)$ in some cases beyond the ones given in the above theorem. As we noted earlier,
if $f$ is a nontrivial  anti-ferromagnetic monotone \emph{symmetric} function, 
the complexity of $\NCSP(f)$ can be determined, but there is no known closed form that indicates when the counting CSP has an FPRAS and when it cannot have an FPRAS unless $\NP=\RP$, cf.\ Theorems~\ref{thm:uniqueness_FPTAS} and \ref{thm:non-uniqueness_hard} (from \cite{Li11:correlation}).

\section{Permissive unaries from pinning}

In this section, we consider the following question: Given a set of functions $\cF$, when does $\Fzo$ contain both a strictly increasing 
permissive unary function
and a strictly decreasing permissive unary function?
In Section~\ref{sec:4.1}
we prove Theorem~\ref{thm:set_pinning} which
shows that if $\Fzo$ does not
contain these, then 
$\cF$ satisfies (at least) one of three  specified properties.
In Section~\ref{sec:4.2}
the goal is to identify a large functional clone that does not
contain 
both a strictly increasing 
permissive unary function
and a strictly decreasing permissive unary function.
The clone~$\MON$, from \cite{bulatov_functional_2017}
fits the bill, 
but in some sense it is a trivial solution since 
it does not contain both~$\delta_0$ and~$\delta_1$.
Theorem~\ref{thm:mon-pm} identifies a clone that is
strictly larger than~$\MON$ and contains $\delta_0$ and~$\delta_1$
but still does not contain 
both a strictly increasing 
permissive unary function
and a strictly decreasing permissive unary function.

\subsection{Condition for having both kinds of unaries}\label{sec:4.1}

The following theorem shows that, unless each element of a set of functions $\cF$ satisfies certain monotonicity conditions on its support, the clone $\Fzo$ contains both a strictly increasing and a strictly decreasing permissive unary function.
Recall from Section~\ref{s:definitions} that $\bar{f}$ denotes the bit-flip of a function $f$, $\bar{\cF}$ the set containing the bit-flips of all the functions in $\cF$, and a pure function is a constant times a relation. 
Furthermore, recall from Section~\ref{s:definitions} that $f\in\cB$ is monotone on its support if for any $\ba,\bb\in R_f$ such that $\ba\le\bb$ it holds $f(\ba)\le f(\bb)$.

\begin{thm}\label{thm:set_pinning}
 Let $\cF$ be a set of functions in $\cB$.
 Then at least one of the following is true:
 \begin{enumerate}
  \item\label{p:all_pure} Every function in $\cF$ is pure.
  \item\label{p:all_monotone} Every function in $\cF$ has the following properties
   \begin{itemize}
    \item it is monotone on its support, and
    \item its support is closed under $\vee$.
   \end{itemize}
   Furthermore, there is some function $f\in\cF$ such that $\bar{f}$ is not monotone on its support.
  \item\label{p:all_antimonotone} $\bar{\cF}$ satisfies Property~\ref{p:all_monotone}.
  \item\label{p:all_up_down} $\ang{\cF,\delta_0,\delta_1}\cap\cB^{<}_1$ and $\ang{\cF,\delta_0,\delta_1}\cap\cB^{>}_1$ are both non-empty.
 \end{enumerate}
\end{thm}
\begin{proof}
 We distinguish cases according to the monotonicity properties of functions in $\cF$.
 
 \textbf{Case 1}. Suppose there exists a function $f\in\cF$ which is not monotone on its support and a function $g\in\cF$ such that $\bar{g}$ is not monotone on its support (these may be the same function).
 
 As $f$ is not monotone on its support, there must be a pair $\ba,\bb\in R_f$ such that $\ba\leq\bb$ and $f(\ba)>f(\bb)>0$.
 By pinning in all places where $\ba$ and $\bb$ agree, we can obtain a function $f'\in\ang{f,\delta_0,\delta_1}$ satisfying $f'(0\zd 0)>f'(1\zd 1)>0$.
 Then $\downf(x):=f'(x\zd x)\in\ang{f,\delta_0,\delta_1}\cap\cB^{>}_1$.
 Similarly, as $\bar{g}$ is not monotone on its support, there must be a pair $\bc,\bd\in R_g$ such that $\bc\leq\bd$ and $0<g(\bc)<g(\bd)$.
 By pinning, we obtain a function $g'\in\ang{g,\delta_0,\delta_1}$ satisfying $0<g'(0\zd 0)<g'(1\zd 1)$, and thus a unary function $\upf(x):=g'(x\zd x)$, which is in $\ang{g,\delta_0,\delta_1}\cap\cB^{<}_1$.
 Hence, $\cF$ satisfies Property~\ref{p:all_up_down}.
 
 \textbf{Case 2}. Suppose all functions in $\cF$ are monotone on their support and there exists a function $f\in\cF$ such that $\bar{f}$ is not monotone on its support.
 
 One possibility is that the supports of all functions in $\cF$ are all closed under $\vee$. In this case, $\cF$ satisfies Property~\ref{p:all_monotone}.
 Otherwise, there exists a function $g\in\cF$ whose support is not closed under $\vee$, that is, there are $\ba,\bb\in R_g$ such that $\ba\vee\bb\notin R_g$.
 Again, $g$ may be the same as $f$, or they may be distinct.
 
 Since $\bar{f}$ is not monotone on its support, there exist $\bc,\bd\in R_f$ such that $\bc\leq\bd$ and $0<f(\bc)<f(\bd)$.
 By pinning and identification of variables, we can therefore realise a permissive unary strictly increasing function $\upf\in\ang{f,\delta_0,\delta_1}$ as in Case 1.
 
 We will now show that we can also realise a permissive unary strictly decreasing function.
 
 By pinning in all places where $\ba$ and $\bb$ agree, we may assume that $\ba\vee\bb=(1\zd 1)$ and $\ba\wedge\bb=(0\zd 0)$.
 Furthermore, by identifying all pairs $i,j$ of variables such that $\ba[i]=\ba[j]$ (and thus $\bb[i]=\bb[j]$), we obtain a function $h(x,y)\in\ang{g,\delta_0,\delta_1}$ such that $h(0,1),h(1,0)\neq 0$ but $h(1,1)=0$.
 Let $a=h(0,0)$, $b=h(0,1)$, $c=h(1,0)$, where $b,c>0$ and $a\geq 0$, and define
 \[
  h'(x,y) := h(x,y) h(y,x) = \begin{pmatrix}a^2 & bc \\ bc & 0 \end{pmatrix}.
 \]
 Let $u_0:=\upf(0)$ and $u_1:=\upf(1)$; these values satisfy $0<u_0<u_1$.
 Now,
 \[
  \sum_{y\in\{0,1\}} h'(0,y) \upf(y) = a^2u_0+bcu_1 > bcu_0 = \sum_{y\in\{0,1\}} h'(1,y) \upf(y),
 \]
 so $\downf(x) := \sum_y h'(x,y) \upf(y)$ is strictly decreasing.
 It is also permissive since $bc\neq 0$ and $u_0>0$, hence $\cF$ satisfies Property~\ref{p:all_up_down}.

 \textbf{Case 3}.
 Suppose that for all functions $g\in \cF$, $\bar{g}$ is monotone on its support and there exists a function $f\in \cF$ such that $f$ is not monotone on its support.

 The first step in Case~3 is to note the following equivalent formulation: 
  All functions in $\bar{\cF}$ are monotone on their support and there exists a function $f\in\bar{\cF}$ such that $\bar{f}$ is not monotone on its support. 
 
 Then, we can apply the argument from Case~2 to $\bar{\cF}$ to find that either $\cF$ satisfies Property~\ref{p:all_antimonotone}, or $\ang{\bar{\cF},\delta_0,\delta_1}\cap\cB^{<}_1$ and $\ang{\bar{\cF},\delta_0,\delta_1}\cap\cB^{>}_1$ are both non-empty.
 Note that the bit-flip operation is  its own inverse and that furthermore $\bar{\delta_0}=\delta_1$ and $\bar{\delta_1}=\delta_0$.
 Thus, if there is a permissive strictly increasing function $\upf\in\ang{\bar{\cF},\delta_0,\delta_1}\cap\cB^{<}_1$, then $\overline{\upf}\in\ang{{\cF},\delta_0,\delta_1}$.
 But $0<\upf(0)<\upf(1)$ implies $0<\overline{\upf}(1)<\overline{\upf}(0)$, so $\overline{\upf}\in\cB^{>}_1$.
 Similarly, if $\downf\in\ang{\bar{\cF},\delta_0,\delta_1}\cap\cB^{>}_1$ then $\overline{\downf}\in\ang{{\cF},\delta_0,\delta_1}\cap\cB^{<}_1$.
 Hence, if $\ang{\bar{\cF},\delta_0,\delta_1}\cap\cB^{<}_1$ and $\ang{\bar{\cF},\delta_0,\delta_1}\cap\cB^{>}_1$ are both non-empty, then $\cF$ satisfies Property~\ref{p:all_up_down}.
 
 \textbf{Case 4}. For all functions $f\in\cF$, both $f$ and $\bar{f}$ are monotone on their support.
 
 One possibility is that  all functions in $\cF$ are pure. In this case, Property~\ref{p:all_pure} is satisfied.
 Otherwise, there exists a function $g\in\cF$ whose range contains at least two non-zero values.
  In this case, we show that there are both a unary permissive strictly increasing and a unary permissive strictly decreasing function in $\ang{g,\delta_0,\delta_1}$.
 
 As $g$ is not pure, we can find $\ba,\bb\in R_g$ such that $0<g(\ba)<g(\bb)$.
 These two tuples $\ba,\bb$ must then be incomparable, otherwise $g$ and $\bar{g}$ could not both be monotone on their support.
 Pin in all places where $\ba$ and $\bb$ agree, then identify all pairs of inputs $i$, $j$ such that $\ba[i]=\ba[j]$ and $\bb[i]=\bb[j]$ to obtain a binary function $h(x,y)$.
 Without loss of generality, we may suppose this function satisfies $h(0,0)=g(\ba\wedge\bb)$, $h(0,1)=g(\ba)$, $h(1,0)=g(\bb)$, and $h(1,1)=g(\ba\vee\bb)$ (otherwise replace $h(x,y)$ with $h(y,x)$).
 
 Now, $g$ and $\bar{g}$ both being monotone on their support also requires that there be no $\bc\in R_g$ which satisfies $\bc\leq\ba\wedge\bb$ or $\bc\geq\ba\vee\bb$.
 In particular, $\ba\wedge\bb$ and $\ba\vee\bb$ cannot be in the support of $g$.
 Hence $h(0,0)=h(1,1)=0$ and $g$ is a weighted disequality function.
 We assumed $0<g(\ba)<g(\bb)$, which implies $0<h(0,1)<h(1,0)$.
 
 Let $\downf(x) := \sum_y h(x,y)h^2(y,x)$ and $\upf(x) := \sum_y h^2(x,y)h(y,x)$, then 
 \begin{align*}
  \downf(0) = h(0,1)h^2(1,0) = g(\ba)g^2(\bb) &> g(\bb)g^2(\ba) = h(1,0)h^2(0,1) = \downf(1) \\
  \upf(0) = h^2(0,1)h(1,0) = g^2(\ba)g(\bb) &< g^2(\bb)g(\ba) = h^2(1,0)h(0,1) = \upf(1).
 \end{align*}
 Both functions are permissive and in $\ang{g,\delta_0,\delta_1}$.
 Therefore, Property~\ref{p:all_up_down} holds.
 
 This concludes the analysis of all cases.
\end{proof}

\begin{rem}
 Properties~\ref{p:all_pure}, \ref{p:all_monotone}, and \ref{p:all_antimonotone} of Theorem~\ref{thm:set_pinning} are disjoint.
 
 Properties~\ref{p:all_monotone} and \ref{p:all_antimonotone} being disjoint is immediate.
 For Properties~\ref{p:all_pure} and \ref{p:all_monotone}, note that if $f$ is a pure function, then both $f$ and $\bar{f}$ are monotone on their support.
 Thus, Property~\ref{p:all_pure} is disjoint from Property~\ref{p:all_monotone}.
 If all functions in $\cF$ are pure, then all functions in $\bar{\cF}$ are pure, so an analogous argument applies for Property~\ref{p:all_pure} and Property~\ref{p:all_antimonotone}.
\end{rem}

\begin{rem}
 If $\cF$ satisfies Property~\ref{p:all_pure} of Theorem~\ref{thm:set_pinning}, then the complexity of $\NCSP(\cF,\delta_0,\delta_1)$ can be determined  
 via the approximation trichotomy for Boolean \#CSP in Theorem~\ref{thm:complexity_Boolean} from \cite{Dyer10:approximation}.
 If $\cF$ satisfies Property~\ref{p:all_up_down}, then the complexity of $\NCSP(\cF,\delta_0,\delta_1)$ is determined by Lemma~\ref{lem:clone_csp} and Theorem~\ref{thm:up-down}.
 So if the complexity of $\NCSP(\cF,\delta_0,\delta_1)$
 is unresolved 
 (up to the granularity of Theorem~\ref{thm:up-down})
 then $\cF$ or $\bar{\cF}$ satisfies Property~2.
\end{rem}

\subsection{Condition for the absence of permissive strictly decreasing functions}
\label{sec:4.2}

In  this section
the goal is to identify a large functional clone that does not
contain both types of
permissive unary functions (the two types being strictly increasing and strictly decreasing).

We start by defining the clone $\MON$,  
using definitions from \cite[Section~10]{bulatov_functional_2017}.
Given a $k$-ary function $f$, let $\sim_f$ be the equivalence relation on $[k]$ given by $i\sim_f j$ if, for every 
$\ba\in\{0,1\}^k$, $f(\ba)\neq 0$ implies $\ba[i]=\ba[j]$.
If $\sim_f$ is the equality relation, $f$ is said to be \emph{irredundant}.
By identifying variables in the equivalence classes of $\sim_f$, any function $f$ can be transformed into an irredundant function $f^\dagger$.
Note that any subset of variables of a function $f$ can be identified using the closure operations of a functional clone (cf.\ Lemma~\ref{lem:closure}), whether or not they are in the same equivalence class.

Recall from Section~\ref{s:definitions} that
a function $g\in\cB_k$ is \emph{monotone} if for any $\ba,\bb\in\{0,1\}^k$ with $\ba\leq\bb$ we have $g(\ba)\leq g(\bb)$.
The function $f$ is in $\MON$ if $f^{\dagger}$ is monotone.  
Bulatov et al.~\cite[Theorem 62]{bulatov_functional_2017} 
have shown that $\MON$ is a functional clone.\footnote{Bulatov et
al.~\cite{bulatov_functional_2017} define $\MON$ in a slightly different but
equivalent way.}
 It does not contain a strictly decreasing permissive unary function, so in some sense
it is a solution to the problem stated in the previous paragraph.
However, it also does not contain $\delta_0$,
so, in some sense, it is a trivial solution.

Theorem~\ref{thm:mon-pm} below identifies a clone that is
strictly larger than~$\MON$ and contains $\delta_0$ and~$\delta_1$
but still does not contain 
both types of
permissive unary functions.

\newcommand{\statemonpm}{
 There is a functional clone $\PM$ containing $\delta_0$ and $\delta_1$ such that $\MON$ is a strict subset of $\PM$ and $\PM$ contains no strictly decreasing permissive unary functions.
}
\begin{thm}\label{thm:mon-pm}
 \statemonpm
\end{thm}

 In order to define $\PM$, we use the following definition.

\begin{dfn}\label{dfn:pin-monotone}
 A function $f\in\cB$ is \emph{pin-monotone} if the following condition holds for $f^\dagger$.
 Let $k=\ari(f^\dagger)$ and suppose that $\ba,\bb\in\{0,1\}^k$ satisfy
 \begin{itemize}
  \item $\ba[j]=0$ and $\bb[j]=1$ for some $j\in[k]$,
  \item $\ba[i]=\bb[i]$ for all $i\in[k]\setminus\{j\}$, and
  \item $f^\dagger(\ba)>f^\dagger(\bb)$,
 \end{itemize}
 then
 \begin{enumerate}
  \item $f^\dagger(\bb)=0$, and furthermore
  \item $f^\dagger(\bc)=0$ for every $\bc\in\{0,1\}^k$ such that $\bc[j]=1$.
 \end{enumerate}
 Let $\PM$ be the set of all pin-monotone functions.
\end{dfn}

In other words, $f$ is pin-monotone if, for every pair of bit strings $\ba$ and $\bb$ that differ only in one bit, with $\ba\leq\bb$ and $f^\dagger(\ba)>f^\dagger(\bb)$, this implies that $f^\dagger(\bb)=0$ and furthermore implies that $f^\dagger(\bc)=0$ for all input bit strings $\bc$ that agree with $\bb$ on the bit where it differs from $\ba$.

\begin{lem}\label{lem:pinning_PM}
 The pinning functions $\delta_0$ and $\delta_1$ are pin-monotone.
\end{lem}
\begin{proof}
Since they have arity~$1$,  both $\delta_0$ and $\delta_1$ are irredundant.
 Then, $\delta_1$ is pin-monotone because for any $\ba,\bb\in\{0,1\}$,  $\ba\leq\bb$ implies $\delta_1(\ba)\leq\delta_1(\bb)$.
 
 For $\delta_0$, there are bit strings $\ba,\bb$ satisfying the three conditions of Definition~\ref{dfn:pin-monotone}, namely $\ba=0$ and $\bb=1$.
 But then the two implications required by that definition are trivially satisfied, since $\delta_0(1)=0$.
 Thus, $\delta_0$ is also pin-monotone.
\end{proof}

\begin{lem}\label{lem:mon_ssneq_PM}
$\MON\subsetneq\PM$. 
\end{lem}
\begin{proof}

To show $\MON\subseteq \PM$,
consider $f\in\MON$ and let $k=\ari(f^\dagger)$.
By the definition of~$\MON$, $f^\dagger$ is monotone.
We will show that $f^\dagger$ is pin-monotone, which implies
that $f$ is also pin-monotone.
 
 To this end, suppose that there are $\ba,\bb\in\{0,1\}^k$ such that $\ba[j]=0$ and $\bb[j]=1$ for some $j\in[k]$, 
 and $\ba[i]=\bb[i]$ for all $i\in[k]\setminus\{j\}$.
 Then by monotonicity of $f^\dagger$, $f^\dagger(\ba)\leq f^\dagger(\bb)$.
 Hence $f^\dagger$ is trivially pin-monotone. We conclude that $f$ is pin-monotone,
 so, since $f$ was an arbitrary function in~$\MON$,
  we conclude that $\MON\sse\PM$.
 
 To show that the inclusion is strict,  note that $\delta_0$ is not monotone because $\delta_0(0)>\delta_0(1)$, yet it is pin-monotone by Lemma~\ref{lem:pinning_PM}.
 Therefore, $\MON\subsetneq\PM$.
\end{proof}

\begin{lem}\label{lem:PM_clone}
 The set $\PM$ is a functional clone.
\end{lem}
\begin{proof}
 The irredundant version of $\EQ$ is the constant function $\EQ^\dagger(x)=1$ for $x\in\{0,1\}$.
 Hence $\EQ$ is pin-monotone, as there is no pair of bit strings $\ba,\bb$ differing in one bit such that $\EQ^\dagger(\ba)>\EQ^\dagger(\bb)$.
 
 To show that $\PM$ is a functional clone, it remains to show that the set is closed under permutation of arguments, introduction of fictitious arguments, product, and summation (see Lemma~\ref{lem:closure}).
 \begin{itemize}
  \item It is straightforward to see that permuting arguments does not destroy the property of being pin-monotone, so $\PM$ is closed under permutation of arguments.
  
  \item Consider the effect of introducing a fictitious argument: let $h(\bx,y)=f(\bx)$, where $f\in\PM$.
   The fictitious argument is in an equivalence class of its own since $h(\bx,0)=h(\bx,1)$ for all $\bx\in\{0,1\}^{\ari(f)}$.
   Hence $h^\dagger(\bx,y)=f^\dagger(\bx)$.
   Let $k=\ari(f^\dagger)$.
   We look at input bit strings to $h^\dagger$ that differ in exactly one bit and distinguish cases according to whether that single bit is the fictitious bit or not.
   
   \textbf{Case 1}. Suppose there exist $\ba,\bb\in\{0,1\}^{k+1}$ such that $\ba[j]=0$ and $\bb[j]=1$ for some $j\in[k]$, $\ba[i]=\bb[i]$ for all $i\in[k+1]\setminus\{j\}$, and $h^\dagger(\ba)>h^\dagger(\bb)$.
   Let $\ba',\bb'$ be the first $k$ bits of $\ba,\bb$, respectively.
   Then $f^\dagger(\ba')>f^\dagger(\bb')$.
   Thus, by pin-monotonicity of $f$, $f^\dagger(\bb')=0$ and $f^\dagger(\bc')=0$ for all $\bc'\in\{0,1\}^k$ such that $\bc'[j]=1$.
   This implies $h^\dagger(\bb)=0$ and $h^\dagger(\bc)=0$ for all $\bc\in\{0,1\}^{k+1}$ such that $\bc[j]=1$.
   
   \textbf{Case 2}. If $\ba,\bb\in\{0,1\}^{k+1}$ such that $\ba[k+1]=0$ and $\bb[k+1]=1$ and $\ba[i]=\bb[i]$ for all $i\in[k]$, then $h^\dagger(\ba)=h^\dagger(\bb)$ by definition.
   
   Hence $h$ is pin-monotone, which implies that the set of pin-monotone functions is closed under introduction of fictitious arguments.
   
  \item To show closure under taking products, let $h(\bx)=f(\bx)g(\bx)$, where $f,g\in\PM$.

   If there exist $i,j\in[\ari(f)]$ with $i\neq j$ such that $f(\bx)\neq 0$  implies $\bx[i]=\bx[j]$, then $h(\bx)\neq0$  implies $\bx[i]=\bx[j]$, so $i$ and $j$ need to be identified to make the irredundant function $h^\dagger$.
   The same holds for any pair of variables that are in the same equivalence class for $g$.
   Hence, $h^\dagger(\bx)=f'(\bx)g'(\bx)$, where $f'$ arises from $f^\dagger$ by identifying some subset of the variables (namely, variables that are in the same equivalence class for $g$ but not $f$), and similarly for $g'$ and $g^\dagger$.

   Let $k=\ari(h^\dagger)$ and suppose $\ba,\bb\in\{0,1\}^k$ satisfy $\ba[j]=0$ and $\bb[j]=1$ for some $j\in[k]$, $\ba[i]=\bb[i]$ for all $i\in[k]\setminus\{j\}$, and $h^\dagger(\ba)>h^\dagger(\bb)$.
   Then, $f'(\ba)>f'(\bb)$ or $g'(\ba)>g'(\bb)$ (or both).
   Assume the former; in the other case, the argument is analogous with $g$ in place of $f$.
   
   The function $f'$ arises from $f^\dagger$ by identifying some variables.
   In particular, suppose the $j$-th variable of $f'$ results from identifying variables $\vc{\ell}{m}$ of $f^\dagger$ for some positive integer $m$.
   Let $\ba'$ be the extension of $\ba$ that satisfies $\ba'[p]=\ba'[q]$ for all pairs $p,q\in[\ari(f^\dagger)]$ such that variable $x_p$ was identified with variable $x_q$ in going from $f^\dagger$ to $f'$, and let $\bb'$ be the analogous extension of $\bb$.
   Then $f^\dagger(\ba')=f'(\ba)>f'(\bb)=f^\dagger(\bb')$.
   Furthermore, for all $s\in [m]$, $\ba'[\ell_s]=0$ and $\bb'[\ell_s]=1$.
   Finally, for all $i\in[\ari(f^\dagger)]\setminus\{\vc{\ell}{m}\}$, $\ba'[i]=\bb'[i]$.
   
   To use the pin-monotonicity property, we need two input bit strings for $f^\dagger$ that differ in exactly one bit.
   If $m=1$, then $\ba'$ and $\bb'$ are such a pair.
   Otherwise, we will identify two suitable bit strings from a sequence of $(m+1)$ bit strings, whose first member is $\ba'$, whose last member is $\bb'$, and where neighbouring members differ in exactly one bit.
   Formally, define the sequence $\ba_0,\vc{\ba}{m}$ as follows.
   Let $\ba_0=\ba'$ and for $s\in [m]$, let $\ba_s$ be the bit string that agrees with $\ba'$ except on the bits with indices $\vc{\ell}{s}$.
   Then, $\ba_m=\bb'$ and furthermore, for each $s\in [m]$, $\ba_{s-1}$ and $\ba_s$ differ in exactly the bit with index $\ell_s$.
   In other words, $\ba_{s-1}[\ell_{s}]=0$, $\ba_s[\ell_{s}]=1$ and $\ba_{s-1}[i]=\ba_s[i]$ for $i\in[\ari(f^\dagger)]\setminus\{\ell_s\}$.
   
   Since $f^\dagger(\ba_0) > f^\dagger(\ba_m)$, there is an integer $t\in [m]$ such that $f^\dagger(\ba_{t-1})>f^\dagger(\ba_t)$.
   Then, by pin-monotonicity of $f$, we have $f^\dagger(\ba_t)=0$ and $f^\dagger(\bc)=0$ for all $\bc\in\{0,1\}^{\ari(f^\dagger)}$ with $\bc[\ell_t]=1$.
   In particular, $\bb'[\ell_t]=1$, so we have $f^\dagger(\bb')=0$ and thus $f'(\bb)=0$ and $h^\dagger(\bb)=f'(\bb)g'(\bb)=0$.
   
   It remains to show that $h^\dagger(\bd)=0$ for all $\bd\in\{0,1\}^k$ that agree with $\bb$ on the $j$-th bit.
   Suppose $\bd\in\{0,1\}^k$ with $\bd[j]=1$.
   Let $\bd'$ be the extension of $\bd$ that satisfies $\bd'[p]=\bd'[q]$ for all pairs $p,q\in[\ari(f^\dagger)]$ such that variable $x_p$ was identified with variable $x_q$ in going from $f^\dagger$ to $f'$.
   Then $\bd'[\ell_t]=1$ since the variables with indices $\vc{\ell}{m}$ are identified to form $x_j$ in going from $f^\dagger$ to $f'$.
   Thus $f^\dagger(\bd')=0$, which implies that $f'(\bd)=0$ and $h^\dagger(\bd)=f'(\bd)g'(\bd)=0$.
   Hence, $h$ is pin-monotone and $\PM$ is closed under products.
  
  \item Finally, consider the effect of summation over a variable: let $h(\bx)=f(\bx,0)+f(\bx,1)$, where $f\in\PM$.
  We distinguish cases according to whether the final input variable of $f$ is equivalent to another input variable.
  
  \textbf{Case 1}. The final input variable of $f$ is in an equivalence class of its own.
  Then $h^\dagger(\bx)=f^\dagger(\bx,0)+f^\dagger(\bx,1)$.
  Let $k=\ari(h^\dagger)$ and suppose there exist $\ba,\bb\in\{0,1\}^{k}$ such that $\ba[j]=0$ and $\bb[j]=1$ for some $j\in[k]$, $\ba[i]=\bb[i]$ for all $i\in[k]\setminus\{j\}$, and $h^\dagger(\ba)>h^\dagger(\bb)$.
  Then either $f^\dagger(\ba,0)>f^\dagger(\bb,0)$ or $f^\dagger(\ba,1)>f^\dagger(\bb,1)$ (or both).
  In either case by pin-monotonicity of $f$ we have $f^\dagger(\bc,d)=0$ for all $d\in\{0,1\}$ and $\bc\in\{0,1\}^{k}$ with $\bc[j]=1$.
  Therefore, $h^\dagger(\bc)=f^\dagger(\bc,0)+f^\dagger(\bc,1)=0$ for any such $\bc$: that is, $h$ is pin-monotone.
  
  \textbf{Case 2}. The final input variable of $f$ is equivalent to the $\ell$-th input variable of $f$, where $\ell\in[\ari(f)-1]$.
  Then $h(\bx)=f(\bx,0)+f(\bx,1)=f(\bx,\bx[\ell])$ for all $\bx\in\{0,1\}^{\ari(h)}$, since $f(\bx,y)$ is 0 unless $y=\bx[\ell]$.
  Furthermore, $h^\dagger(\bx)=f'(\bx,\bx[\ell])$, where $f'$ arises from $f$ by identifying all variables in the same equivalence classes, except the $\ell$-th and the final variable.
  But $f'(\bx,\bx[\ell])=f^\dagger(\bx)$ since identification of two variables that are always equal does not change the value of the function.
  Therefore $h^\dagger(\bx)=f^\dagger(\bx)$, and pin-monotonicity of $f$ immediately implies pin-monotonicity of $h$.
  
  Hence, $\PM$ is closed under summation.
 \end{itemize}

 This concludes the proof.
\end{proof}

We can now prove the central theorem of this section.

\begin{repthm}{thm:mon-pm}
 \statemonpm
\end{repthm}
\begin{proof}
 Recall that $\PM$  is the set of all pin-monotone functions.
 This is shown to be a functional clone in Lemma~\ref{lem:PM_clone}.
 The pinning functions $\delta_0,\delta_1$ are in $\PM$ by Lemma~\ref{lem:pinning_PM}.
 Furthermore, by Lemma~\ref{lem:mon_ssneq_PM}, we have $\MON\subsetneq\PM$.
 
 Now let $f$ be an arbitrary strictly decreasing permissive unary function, i.e.\ $f(0)>f(1)>0$; then $f$ is irredundant.
 The bit strings $\ba=0$, $\bb=1$ satisfy the three conditions of Definition~\ref{dfn:pin-monotone} for $f$.
 Yet $f(\bb)\neq 0$, so $f$ is not pin-monotone.
 Since $f$ was arbitrary, this establishes that $\PM$ does not contain any strictly decreasing permissive unary functions.
\end{proof}

As the pinning functions are contained in $\PM$, we also have the following corollary.

\begin{cor}
 For any $\cF\sse\PM$, $\Fzo$ does not contain any strictly decreasing permissive unary functions.
\end{cor}

\bibliographystyle{plain}
\bibliography{\jobname}

\end{document}